%% file: main.tex
\documentclass[final,leqno,letterpaper]{article}
\input{preamble.tex}
\hypersetup{%
    pdftitle={Porous Convection in the Discrete Exterior Calculus with Geometric Multigrid},
    pdfauthor={Luke Morris, George Rauta, Kevin Carlson, James Fairbanks},
    pdfkeywords={discrete exterior calculus, porous convection, and geometric multigrid}
    }

\title{Porous Convection in the Discrete Exterior Calculus with Geometric Multigrid\thanks{
Received... Accepted... Published online on... Recommended by....
}}

\author{Luke Morris\footnotemark[2], George Rauta\footnotemark[2], Kevin Carlson\footnotemark[3], and James Fairbanks\footnotemark[2]}

\begin{document}

\maketitle

\renewcommand{\thefootnote}{\fnsymbol{footnote}}

\footnotetext[2]{Department of Computer and Information Science and Engineering, University of Florida, Gainesville, FL 32611 USA (luke.morris@ufl.edu, grauta@ufl.edu, fairbanksj@ufl.edu)}
\footnotetext[3]{Topos Institute, Berkeley, CA 94704 USA (kevin@topos.institute)}

\begin{abstract}
The discrete exterior calculus (DEC) defines a family of discretized differential operators which preserve certain desirable properties from the exterior calculus. We formulate and solve the porous convection equations in the DEC via the Decapodes.jl embedded domain-specific language (eDSL) for multiphysics problems discretized via CombinatorialSpaces.jl. CombinatorialSpaces.jl is an open-source Julia library which implements the DEC over simplicial complexes, and now offers a geometric multigrid solver over maps between subdivided simplicial complexes. We demonstrate numerical results of multigrid solvers for the Poisson problem and porous convection problem, both as a standalone solver and as a preconditioner for open-source Julia iterative methods libraries.
\end{abstract}

\input{sec/introduction.tex}

\input{sec/methodology.tex}

\input{sec/results.tex}

\input{sec/conclusions.tex}

\input{sec/acknowledgments.tex}

\bibliographystyle{siam}
\bibliography{references.bib}

\end{document}

%% file: preamble.tex
\usepackage{fullpage}
\usepackage{amsthm}
\usepackage{graphicx,amsmath,amsfonts}
\usepackage{amssymb}
\usepackage{hyperref}
\usepackage{subcaption}
\newtheorem{definition}{Definition}
\newtheorem{proposition}{Proposition}
\usepackage{tikz}
\usepackage{tikz-cd}
\usepackage{quiver}
\usepackage{graphicx}
\usepackage{soul}
\usepackage{cleveref}
\usepackage{lineno}

\usepackage[inline]{asymptote}
\usepackage{booktabs}
\newcommand{\uflat}{\mathbf{u}^\flat}
\newcommand{\extd}{\mathbf{d}}
\newcommand{\extcodif}{\mathbf{\delta}}
\newcommand{\Lie}{\mathcal{L}}
\newcommand{\sds}{\star \, \extd \, \! \star}

\usetikzlibrary{matrix,arrows}


%% file: sec/introduction.tex
\section{Introduction}

The discrete exterior calculus (DEC) is a discretization of the exterior calculus which is similar to the finite volume method (FVM) over a primal-dual mesh setup. Given a mesh, the DEC provides constructions of discrete differential operators as matrix-vector and kernel-based operations. These operators and the spaces they act on can be visualized in a diagram known as the de Rham complex, as shown in Figure \ref{fig:dRC}. The DEC is valued for its structure-preserving (mimetic) qualities, chiefly that the curl of the gradient is 0 is enforced at the discrete level i.e., $\extd \extd = 0$ in exterior calculus notation. The DEC was studied in its own right in the work of Hirani~\cite{hirani_discrete_2003}, Desbrun~\cite{desbrun_discrete_2005}, Leok~\cite{leok_foundations_2004}, and others, and the study of discrete differential forms\footnote[7]{For introductory purposes, discrete differential forms are simply functions evaluated over the points, edges or triangles, of a mesh.} can be recognized in particular cases going back to at least 1977 in the context of the Finite Integration Technique~\cite{clemens_discrete_2001, weiland_discretization_1977}. Further, the DEC itself can be found in the development of TRiSK-type schemes~\cite{eldred_interpretation_2022}. These instances of parallel development are anecdotal of the expressive power of discrete versions of the exterior calculus as a unifying theoretical framework for model development in the applied sciences. The exterior calculus has also been found useful in formalizing the finite element method via the finite element exterior calculus (FEEC). FEEC was originally explored in the work of Arnold, Falk, and Winther in myriad works~\cite{arnold_finite_2006, arnold_differential_2006, arnold_finite_2010}, and various efforts on improving the FEEC framework have taken place including $hp$-refinement~\cite{gates_hp-hierarchical_2021} and hybridized methods~\cite{awanou_hybridization_2023}. Improvements to the generic DEC framework are plentiful since its introduction in the early 2000's. These improvements include, for example, the introduction of upwinding methods~\cite{heumann_extrusion_2008, mullen_discrete_2011} and adjustments for simulations over bounded domains~\cite{eldred_structure-preserving_2021}. Most relevant to the present work are the application of algebraic multigrid (AMG) to the DEC in the work of Bell et al.~\cite{bell_algebraic_2008-1, bell_algebraic_2008}, and the ``subdivision exterior calculus" of de Goes et al.~\cite{de_goes_subdivision_2016}, which interprets subdivision surfaces from a DEC perspective. Also relevant is the unpublished \texttt{pycomplex} package which implements multigrid restriction and prolongation operations~\cite{hoogendoorn_eelcohoogendoornpycomplex_2025} over spheres and regular grids. Outside of the field of exterior calculus, the approach to be actualized by the simulations in this current paper is most similar to nested-mesh subdivision schemes~\cite{mavriplis_mutigrid_1995}, though the geometric maps developed in this paper are designed to capture more general maps between domains. In terms of particular (multi)-physics simulations, the DEC has been applied to the incompressible Navier-Stokes equations~\cite{mohamed_discrete_2016}, ideal magnetohydrodynamics (MHD) equations~\cite{kraus_variational_2018}, and rotating shallow water equations~\cite{eldred_interpretation_2022} to name a few.

\begin{figure}[hbtp]
    \centering
    \scalebox{1.3}{\input{img/deRham2D.tex}}
    \caption{The 2D de Rham complex relates differential forms on a primal (top row) and dual (bottom row) mesh using the exterior derivative operation and the Hodge star operation. For example, $d_0$ takes primal 0-forms to primal 1-forms, and $\star_0$ takes primal 0-forms to dual 2-forms. The qualitative descriptions of these operations can be found in Section 3.}
    \label{fig:dRC}
\end{figure}

\begin{figure}[hbtp]
    \centering
    \scalebox{0.85}{\input{img/dRC_Stratified_AllMaps}}
    \caption{The 2D de Rham complex of geometrically related spaces are connected by interpolation and restriction maps between a (primal) domain and two successive refinements of that domain - one leg of a V-cycle. Note that only maps between $\Omega(X_0)$ (primal 0-forms) are necessary and implemented for solves of the usual Poisson problem.}
    \label{fig:dRC_Stratified}
\end{figure}

The abstract perspective on space and shape taught by algebraic topology has found application in fields of statistics, data science, and machine learning through topological data analysis (TDA)~\cite{dey_topological_2022}. In the TDA paradigm, data sets are represented as point clouds in Euclidean space and then converted to topological spaces as the geometric realizations of simplicial complexes. Because the spatial scale of data analysis is a free parameter, simplicial complex filtrations such as Vietoris-Rips or Cech complexes are used to create a sequence of simplicial complexes that approximate the data. As the spatial resolution changes, these complexes fill in with simplices. The tools of algebraic topology are used to understand changes in the shape of the data as you change the spatial resolution of your analysis. By focusing on the topology and geometry of simplicial complexes, TDA research has developed powerful tools for data analysis~\cite{nanda_computational_2021}.

In the DEC, a discrete de Rham complex captures the space of differential forms over a discrete space. In this setting we are looking at piecewise linear differential forms over the geometric realization of a simplicial complex. The relevant differential operators of exterior calculus have been discretized and thus physics modeling can be done entirely over the discrete domain without reference to the continuous quantities that are being approximated. The musical operators $\sharp$ and $\flat$ transfer data between discrete differential forms and discrete vector fields, allowing DEC models to integrate with traditional vector calculus modeling tools. A benefit of the DEC is that continuous models described in the language of exterior calculus can be directly implemented with a canonical discrete approximation. This canonical discretization can be implemented algorithmically, making it possible to automate the construction of complex simulators directly from a mathematical description of the model~\cite{morris_decapodes_2024}. Because of how the discrete operators are defined relative to a simplicial mesh, DEC implementations usually track conserved quantities -- including mass, momentum, and energy -- to machine precision without special effort~\cite{wang_discrete_2023,mohamed_discrete_2016}.  

In this paper, we take the simplicial complex perspective and develop a general geometric multigrid (GMG) formulation based on barycentric coordinates and integrate it with a DEC-based simulation of complex fluid dynamics on general triangulated meshes. We are interested in directly encoding the V-cycles, W-cycles, and so on from geometric multigrid in the DEC, keeping in mind that multigrid methods on unstructured meshes have been studied independently of this framework~\cite{mavriplis_mutigrid_1995}. This is in contrast to the work of subdivision exterior calculus~\cite{de_goes_subdivision_2016}, for example, in which similarities to multigrid are found post facto. For the present study, we restrict ourselves to familiar scalar fields (discrete differential 0-forms). We are interested in using geometric multigrid methods to improve existing DEC simulation frameworks for practical computational fluid dynamics (CFD) problems. So, we restrict ourselves to the study of the maps between primal 0-forms (scalar fields) between refinements of a domain, since this enables us to improve extant simulations of the porous convection equations as in Section 3\footnote{Although the vector-Poisson problem, for example, may be solved by using multigrid methods directly on higher-dimensional differential forms, such efforts are tangential to typical fluid flow problems.}.
To the best of our knowledge, this particular benchmark is a contribution in itself, this being the first implementation of such a porous convection simulation in the discrete exterior calculus.
The restriction and prolongation maps used in multigrid can be neatly formulated in the language of simplicial complexes. We use geometric maps between simplicial complexes to represent the relationships between meshes and define the induced restriction and prolongation operators. Particularly important are the geometric maps that encode the relation between a mesh and a refinement of that mesh. These are the maps necessary to set up a GMG scheme. Although algebraic multigrid methods may demonstrate superior convergence in particular cases~\cite{wu_analysis_2006}, their ``blackbox" nature happens to side-step the geometric maps that we seek for a unifying approach of multiscale, multidomain, and multiphysics models. Further, we believe that an explicit study of GMG in the DEC will complement prior works such as the study of AMG in Bell et al.~\cite{bell_algebraic_2008-1, bell_algebraic_2008}.

A main motivation of this work is to show that an algebraic topology perspective on PDE solvers can be used to derive elegant representations of sophisticated numerical methods. The utility of studying general geometric maps, rather than skipping directly to refinements alone, is to build a framework for using multigrid with multidomain (in the vein of (optimized) Schwarz decomposition methods) and multiphysics models. In order to define a multidomain map, you need to define operators that transfer data between the component domains. These will have the same form as the restriction maps used in our multigrid implementation described below. The general algebraic topology perspective will enable this multigrid method to integrate with multidomain methods in future work. So, we develop here a category theoretic formalism for multiscale methods, to be reused as the basis for future development of multidomain methods.

We start by considering geometric morphisms of simplicial complexes and end with numerical experiments on porous media flow problems which incorporate the resulting multigrid solver. This combination of geometric morphisms between domains and discrete differential operators defined on those domains is shown in Figure \ref{fig:dRC_Stratified}. Here, the de Rham complex associated with each discretization of a domain is shown, and interpolation and restriction maps are created between the primal differential forms. Each node in this diagram represents a space of differential forms of particular degree, and each arrow (morphism) is an operator that acts on this space of differential forms.\footnote{Identity morphisms that map a space back to itself are elided.} We will see in Section \ref{sec:results} that the discrete differential operators of the usual de Rham complex (Figure \ref{fig:dRC}) can be composed to create further discrete differential operators, such as the Hodge Laplacian operator: $\Delta_0 = \star_0^{-1} \tilde{\extd_1} \star_1 \extd_0$. We will show in Section \ref{sec:methodology} that GMG methods can be derived by composing operators from a stratification of the de Rham complex and geometric maps between domains, as in Figure \ref{fig:dRC_Stratified}.

This work builds on prior work towards automation in multiphysics simulation using the DEC and the framework of lifting problems from algebraic topology~\cite{patterson_diagrammatic_2023}. This approach used diagrams, opfibrations, and lifting problems to define a compositional framework for multiphysics problem specification. This paper attempts to take this algebraic topology perspective within the DEC further to develop new approaches to compositional multiscale methods relying on subdivision of simplicial complexes and their actions on cochain complexes that are used to represent physical quantities over a space. The next step will be to use this approach in defining multidomain, multiscale, and multiphysics simulations in the DEC. 

We have released an open source implementation of this work and integrated it into the Decapodes framework for multiphysics simulations in the DEC~\cite{morris_decapodes_2024}. This demonstrates the generality of the approach that the DEC and category theoretic methods enable. These simulations are generic over not only well-formed manifolds, but also their refinements.

The structure of this paper is as follows. Section 1 (this section) explains the motivation for this work in the field of the discrete exterior calculus and the research program of multiscale, multidomain, and multiphysics simulations. Section 2 explains the particular methodology by which we introduce GMG into the DEC in an approach inspired by the category of piecewise linear geometric maps and its translation into typical linear algebra. Section 3 details translations of the Poisson problem and a porous convection problem into the DEC. Both of these physics entail a linear solve, in which we employ the DEC GMG solver of Section 2. In Section 4, we make final conclusions, and acknowledgments are to be found in Section 5.

\subsection{Contributions}

The contributions of this paper are as follows:
\begin{enumerate}
    \item A geometric multigrid method for the discrete exterior calculus on simplicial complexes
    \item A discrete exterior calculus formulation of the porous convection problem
\end{enumerate}

We validate these two contributions with numerical simulations. (1) is validated on both the discrete Poisson problem and on (2). The numerical validation of (2) is based on comparison against a baseline direct method and comparable non-GMG iterative methods.

%% file: img/deRham2D.tex
\begin{tikzcd}
	{\Omega_0} && {\Omega_1} && {\Omega_2} \\
	\\
	{\tilde\Omega_{2}} && {\tilde\Omega_{1}} && {\tilde\Omega_{0}}
	\arrow["d_0", from=1-1, to=1-3]
	\arrow["d_1", from=1-3, to=1-5]\
	\arrow["\star_0"', shift right=2, from=1-1, to=3-1]
	\arrow["\star_1"', shift right=2, from=1-3, to=3-3]
	\arrow["\star_2"', shift right=2, from=1-5, to=3-5]
	\arrow["{\tilde d_0}", from=3-5, to=3-3]
	\arrow["{\tilde d_1}", from=3-3, to=3-1]
	\arrow["{\star_0^{-1}}"', shift right=2, from=3-1, to=1-1]
	\arrow["{\star_1^{-1}}"', shift right=1, from=3-3, to=1-3]
	\arrow["{\star_2^{-1}}"', shift right=1, from=3-5, to=1-5]
\end{tikzcd}

%% file: img/dRC_Stratified_AllMaps.tex
\begin{tikzcd}
	{\Omega(X_0)} && {\Omega(X_1)} && {\Omega(X_2)} \\
	\\
	{\tilde{\Omega}(X_2)} && {\tilde{\Omega}(X_1)} && {\tilde{\Omega}(X_0)} \\
	& {\Omega(X_0)} && {\Omega(X_1)} && {\Omega(X_2)} \\
	\\
	& {\tilde{\Omega}(X_2)} && {\tilde{\Omega}(X_1)} && {\tilde{\Omega}(X_0)} \\
	&& {\Omega(X_0)} && {\Omega(X_1)} && {\Omega(X_2)} \\
	\\
	&& {\tilde{\Omega}(X_2)} && {\tilde{\Omega}(X_1)} && {\tilde{\Omega}(X_0)}
	\arrow["{d_0}", from=1-1, to=1-3]
	\arrow["{\star_0}"', from=1-1, to=3-1]
	\arrow["int_0"{description, pos=0.4}, curve={height=-6pt}, dashed, from=1-1, to=4-2]
	\arrow["{d_1}", from=1-3, to=1-5]
	\arrow["{\star_1}"', from=1-3, to=3-3]
	\arrow["int_1"{description, pos=0.4}, curve={height=-6pt}, dashed, from=1-3, to=4-4]
	\arrow["{\star_2}"', from=1-5, to=3-5]
	\arrow["int_2"{description, pos=0.4}, curve={height=-6pt}, dashed, from=1-5, to=4-6]
	\arrow["{\star_0^{-1}}", curve={height=-12pt}, from=3-1, to=1-1]
	\arrow["{\star_1^{-1}}", curve={height=-12pt}, from=3-3, to=1-3]
	\arrow["{\tilde{d}_1}", from=3-3, to=3-1]
	\arrow["{\star_2^{-1}}", curve={height=-12pt}, from=3-5, to=1-5]
	\arrow["{\tilde{d}_0}", from=3-5, to=3-3]
	\arrow["res_0"{description}, curve={height=-6pt}, dashed, from=4-2, to=1-1]
	\arrow["{d_0}", from=4-2, to=4-4]
	\arrow["{\star_0}"', from=4-2, to=6-2]
	\arrow["int_0"{description, pos=0.4}, curve={height=-6pt}, dashed, from=4-2, to=7-3]
	\arrow["res_1"{description}, curve={height=-6pt}, dashed, from=4-4, to=1-3]
	\arrow["{d_1}", from=4-4, to=4-6]
	\arrow["{\star_1}"', from=4-4, to=6-4]
	\arrow["int_1"{description, pos=0.4}, curve={height=-6pt}, dashed, from=4-4, to=7-5]
	\arrow["res_2"{description}, curve={height=-6pt}, dashed, from=4-6, to=1-5]
	\arrow["{\star_2}"', from=4-6, to=6-6]
	\arrow["int_2"{description, pos=0.4}, curve={height=-6pt}, dashed, from=4-6, to=7-7]
	\arrow["{\star_0^{-1}}", curve={height=-12pt}, from=6-2, to=4-2]
	\arrow["{\star_1^{-1}}", curve={height=-12pt}, from=6-4, to=4-4]
	\arrow["{\tilde{d}_1}", from=6-4, to=6-2]
	\arrow["{\star_2^{-1}}", curve={height=-12pt}, from=6-6, to=4-6]
	\arrow["{\tilde{d}_0}", from=6-6, to=6-4]
	\arrow["res_0"{description}, curve={height=-6pt}, dashed, from=7-3, to=4-2]
	\arrow["{d_0}", from=7-3, to=7-5]
	\arrow["{\star_0}"', curve={height=-6pt}, from=7-3, to=9-3]
	\arrow["res_1"{description}, curve={height=-6pt}, dashed, from=7-5, to=4-4]
	\arrow["{d_1}", from=7-5, to=7-7]
	\arrow["{\star_1}"', from=7-5, to=9-5]
	\arrow["res_2"{description}, curve={height=-6pt}, dashed, from=7-7, to=4-6]
	\arrow["{\star_2}"', from=7-7, to=9-7]
	\arrow["{\star_0^{-1}}", curve={height=-6pt}, from=9-3, to=7-3]
	\arrow["{\star_1^{-1}}", curve={height=-12pt}, from=9-5, to=7-5]
	\arrow["{\tilde{d}_1}", from=9-5, to=9-3]
	\arrow["{\star_2^{-1}}", curve={height=-12pt}, from=9-7, to=7-7]
	\arrow["{\tilde{d}_0}", from=9-7, to=9-5]
\end{tikzcd}

%% file: sec/methodology.tex
\section{Multigrid primitives based on geometric maps}
\label{sec:methodology}

Our multigrid implementation aims at the familiar and interpretable behavior of traditional geometric multigrid with a generic interface for arbitrary geometric maps between simplicial meshes. These meshes are stored in a simplicial
set data structure, provided by \texttt{CombinatorialSpaces.jl}, which generates the operators required for a spatially discretized PDE. The information necessary to provide a variety of multigrid methods can be derived from geometric morphisms between simplicial complexes. This includes V-Cycles, W-Cycles, and F-Cycles in any combination. 
Being parameterized by an arbitrary geometric map allows this GMG method to apply to a wide variety of subdivision schemes. This provides flexibility to the implementers of a simulation, while representing geometric information explicitly.

\subsection{Geometric constructions}
In this section, we will answer the questions of how to represent points, how to represent smooth maps, and how to derive operators from these. We will proceed to answer these questions in this order.

\begin{figure}[htbp]
\begin{center}
\includegraphics[width=0.8192\textwidth]{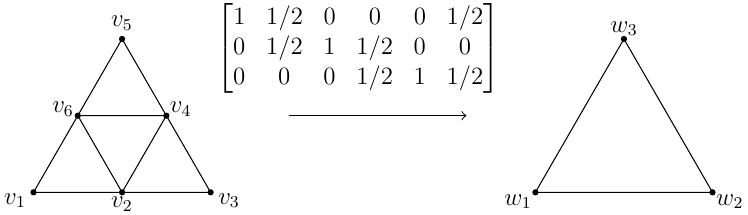}
\end{center}
\caption{An example map from a (finer) subdivision of a mesh to the coarse mesh.}
\label{fig:ex-subdiv}
\end{figure}

A key desideratum for our data structures encoding maps between simplicial
complexes is that they encompass the smooth map from a subdivision of a complex
to the complex itself, as illustrated in Figure \ref{fig:ex-subdiv}. Ordinary maps
of simplicial complexes -- which map vertices to vertices, edges to edges, and so on -- do not permit this flexibility. These would require each 
simplex of the domain to be mapped surjectively onto a simplex of the codomain, which does not capture the information needed to transfer for GMG applications.

Instead, our basic data structure for simplicial complexes may be interpreted as follows.
We give the case for 2-dimensional complexes embedded in 3-dimensional space for explicitness and because it is of the main interest
for our software, but the reader with some experience with simplicial sets will see 
how to generalize to dimension $n.$

\begin{definition}
    A (2-dimensional) \emph{embedded simplicial complex} consists of a \emph{point cloud} 
    $p:V\to \mathbf{R}^3,$ where $V$ is some finite set of \emph{vertices}, together with finite sets
    $E,T$ of edges and triangles and face maps $d_0,d_1:E\to V$ and $d_0,d_1,d_2:T\to E$ such that
    \begin{itemize}
        \item The standard simplicial identities are satisfied for the face maps (so that adjacent edges of each triangle share a common vertex.)
        \item No two simplices in $E$ or in $T$ have the same set of vertices, and no vertex or edge is repeatedly a face of the same higher-dimensional simplex.
        \item Every vertex and edge is in at least one triangle.
    \end{itemize}
\end{definition}

Note that the conditions imposed imply that, mathematically, we could define a simplicial
complex by specifying $V$ together with a downward-closed subset of the powerset $P(V)$ :
in other words, by identifying an $n$-simplex with its $n+1$ vertices; this is the intuition behind the simplex tree data structure.
This mathematical data structure can represent simplicial complexes, but is much less flexible
for generalization and is less closely in keeping with our software.

\begin{figure}[htbp]
    \centering
    \includegraphics[width=0.5\textwidth]{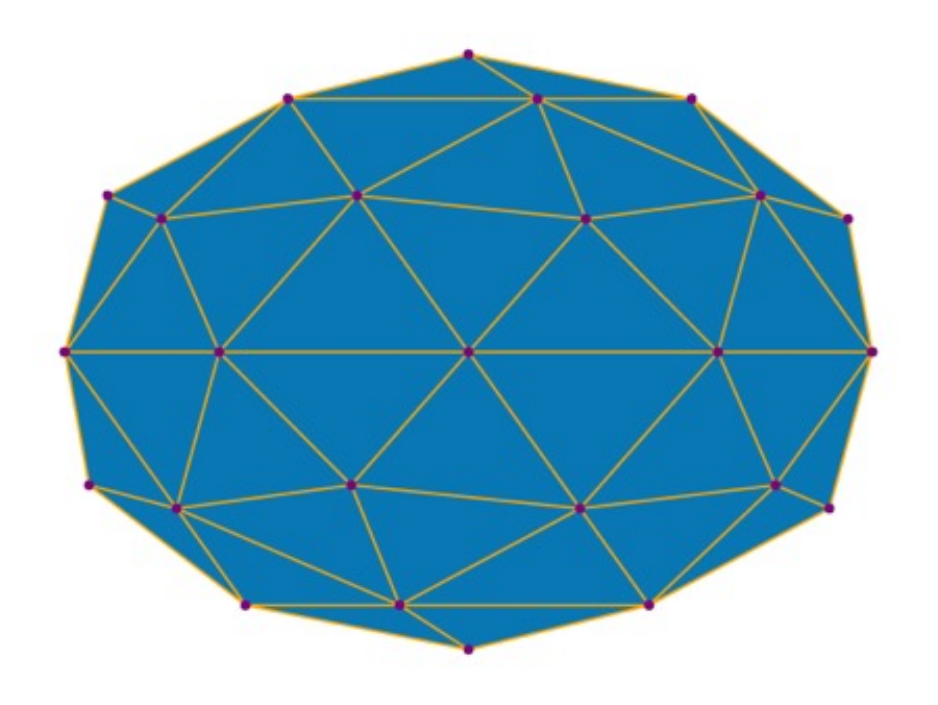}
    \caption{A primal simplicial complex, with primal vertices in purple, primal edges in orange, and primal triangles in blue.}
    \label{fig:methodology-illustration-mesh}
\end{figure}

We implicitly consider the higher-dimensional faces to lie (bi)-linearly interpolated between their vertices, which
allows us to implement the above mathematical definition in a finitary data structure, namely \texttt{EmbeddedDeltaSet2D} 
in \texttt{CombinatorialSpaces.jl}. Figure \ref{fig:methodology-illustration-mesh} illustrates one such delta set. We rely on \texttt{ACSets.jl}~\cite{patterson_categorical_2022} for fast, flexible data structures for
attributed combinatorial objects such as simplicial complexes.
We focus on a particular class of piecewise-linear morphisms of simplicial complexes applicable to our multigrid 
case:

\begin{definition}
    A \emph{geometric map} between embedded simplicial complexes $X,Y$ is a continuous function $f:X\to Y$, affine on each simplex, and such that
    each simplex of $X$ has image entirely contained in some simplex of $Y$.
\end{definition}

The assumption that every positive-dimensional simplex is linearly interpolated among its 
vertices allows us to encode
such a geometric map by simply specifying the image of each vertex of $X$ in $\mathbb{R}^3,$ 
then simply check that, if $\{x_1,\ldots,x_n\}$ are the vertices of some simplex of 
$X$, then the images $f(x_1),\ldots,f(x_n)$ all lie in some single simplex of $Y$. For example, an edge $e$ must have both of its endpoints $\{x_1, x_2\}$ lie within the bounds of the same triangle $t$.

\begin{proposition}
    A geometric map $f:X\to Y$ of embedded simplicial complexes is uniquely determined 
    by its restriction to the set $X_0$ of vertices of $X$, and a function $f_0:X_0\to Y$
    extends to a (unique) geometric map $X\to Y$ if and only if $f_0$ sends the vertices 
    of every simplex of $X$ into some simplex of $Y.$
\end{proposition}
\begin{proof}
    Given $f_0:X_0\to Y$ and a point $x\in X,$ we can express $x$ uniquely as a 
    linear combination $x=a^ix_i$ of the vertices of the minimal simplex in which $x$
    lies. (In fact, the $a_i$ are all in $[0,1]$ and sum to $1$: they are the barycentric
    coordinates of $x$ that we shall make more systematic use of in the next section.)
    Therefore the only possible $f$ extending $f_0$ is defined by 
    $f(x)=a^if_0(x_i),$ which establishes uniqueness. Given a simplex of $X$ with vertices 
    $\{x_1,\ldots,x_n\}$, it is clear that the linear combination above lies in some fixed simplex
    of $Y$ for every choice of coefficients if and only if all the $x_i$ do, since simplices
    are convex. 
\end{proof}

Simplicial complexes (in any dimension) equipped with such geometric maps form a category
$\mathsf{GM}$. One composes these maps by composing them as functions. 
The fact that geometric maps between simplicial complexes form a category justifies using diagrams as a rigorous tool for reasoning about relationships between spaces. These diagrams for reasoning about relationships between spaces are rigorous diagrams analogous to commuting diagrams used in abstract algebra. Such diagrams can be used to specify complex multigrid as discussed in the following section.
A primary feature of $\mathsf{GM}$ is that this category can be represented in linear 
algebra using barycentric coordinates and matrix multiplication, as we now describe.

\subsection{Linear algebraic formulation}
In the previous section, we did not describe how geometric maps are actually encoded,
beyond the observation that they are designed to be specified on vertices alone. 
The most obvious approach would be to encode the image of a vertex $x$ of a simplicial
complex $X$ under a geometric map $f:X\to Y$ by its coordinates in $\mathbb R^3$, 
but this would abandon most of the advantages of our data structures, since validation
would require geometric computations.

Instead, we implement an encoding of geometric maps which is independent of the embedding
and cannot express invalid maps, using \emph{barycentric coordinates}. Recall that, 
if $x,y,z$ are three points in a Euclidean space, then the barycentric coordinates 
$[a,b,c]$ with respect to $x,y,z$ refer simply to the linear combination $ax+by+cz.$ 
The significant case is when $a,b,c\ge 0$ and $a+b+c=1,$ for every point in 
the triangle spanned by $x,y,z$ has unique barycentric coordinates of this 
form, as we illustrate below:

\begin{center}
\includegraphics[width=0.5\textwidth]{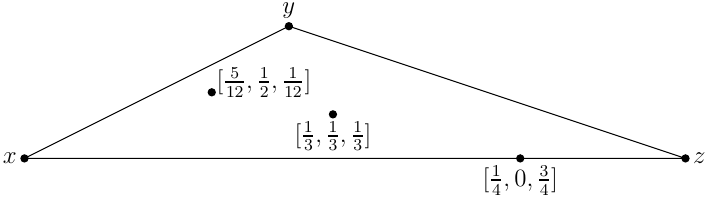}
\end{center}

Mathematically, we can thus encode a geometric map $f:X\to Y$ entirely as a function 
$X_0\to Y_0+Y_1\times \Delta^1_\circ+Y_2\times \Delta^2_\circ,$ (with the obvious modifications 
in different dimensions) where $\Delta^2=\{(a,b,c)\in \mathbb{R}^3\mid a,b,c> 0,a+b+c=1\}$
is the open standard 2-simplex. This function picks for each vertex of $X$ a simplex of $Y$ 
into whose interior it maps, and the barycentric coordinates in that triangle of the intended image. 
As a final step in simplifying our data structures 
we introduce the \emph{global} barycentric coordinates on an 
embedded simplicial complex $X.$ This representation conforms to standard sparse linear algebra
tools for efficient storage.
\begin{definition}
If $X$ is an embedded simplicial complex with $k$ vertices $\vec x=[x_1,\ldots,x_k]$, then a 
global barycentric coordinate is given by an element $\vec{a}\in\mathbb R^k$ such that
$a_i\ge 0,\sum\vec a=1$ and the linear combination $\vec{a}\cdot \vec x$ actually lies in 
$X.$
\end{definition} 

In particular, if $X$ is $2$-dimensional, then in general position at most three of the entries 
of $\vec{a}$ can be nonzero, since a fourth nonzero coordinate would define a point in 
the interior of the tetrahedron spanned by the corresponding four vertices in $X.$ 
The global barycentric coordinates pick out a point
by selecting the simplex $s$ of $X$ corresponding to the nonzero entries of 
$\vec{a},$ then using the ordinary, or local, barycentric coordinates, within $s.$
The important upshot is that we can now define a very convenient encoding of 
geometric maps, which embeds their category into the category of matrices. 

\begin{definition}\label{def:mfunc}
The \emph{matrix} $M(f)$ of a geometric map $f:X\to Y$, where $X$ and $Y$ have 
vertices $\{x_1,\ldots,x_m\}$ and 
$\{y_1,\ldots,y_n\}$, respectively, is the $n\times m$ matrix whose $j$th column is the vector
of global barycentric coordinates of $f(x_j)$ in $Y.$
\end{definition}

This is a good encoding of geometric maps insofar as the composition operations are respected,
which is necessary to use this data structure in the highly compositional multigrid algorithms
that are our target.

\begin{proposition}\label{prop:mfunc}
The operation of taking the matrix defines a functor $M:\mathsf{GM}\to \mathsf{Vect}.$ That is, 
if $X\xrightarrow{f} Y\xrightarrow{g} Z$ is a composable sequence of geometric maps
of simplicial complexes, then $M(g\circ f)=M(g)\cdot M(f),$ while $M(\mathrm{id}_X)=I_n$.
\end{proposition}
\begin{proof}
    Consider a vertex $x\in X_0.$ Denote $f(x)=a^jy_j$ using global barycentric coordinates in 
    $y,$ and denote $g(y_j)=b^k_jz_k$ in the same manner. Then since $g$ is affine on 
    simplices of $Y$ and $f(x)$ lies in some such simplex, we have 
    $(g\circ f)(x)=a^jb^k_jz_k.$ This shows that $x$'s column in $M(g\circ f)$ corresponds 
    to its column in $M(g)\cdot M(f),$ which proves the statement on composition. The 
    identity geometric map sends each vertex $x\in X_0$ to itself with barycentric coordinate $(1,)$.
    The corresponding matrix representation is $I_n$.
\end{proof}

This result lets us compute with geometric maps, albeit a class that models reasonably 
general continuous maps of simplicial complexes, taking direct advantage of fast
linear algebraic algorithms, and independently of the embeddings.
When stored in a compressed sparse column (CSC) matrix, the global barycentric components matrix is capturing the data of the geometric map directly. Each column of this matrix stores the image of a vertex in the domain. It will have at most $k$ nonzeros for a $k$ dimensional domain. The nonzeros in that column will occur in rows $i_1,\dots, i_k$ which are the integers stored in the CSC representation.

The functoriality of $M$ means that any diagram of spaces (Figure \ref{fig:spatial-diagram}) can be converted into a rigorous diagram in the category of vector spaces and linear maps. These diagrams are usually used informally to describe multigrid methods based on their shapes, for example V-Cycles, W-Cycles, etc. By taking diagrams of geometric maps and functorially constructing the corresponding diagram of linear maps, we can use these descriptions as formal specifications of complex cycling schemes. For example what might be verbally described as ``a 2 level V-cycle with binary subdivision followed by a 2 level W-cycle with ternary subdivision'' is formally specified in Figure \ref{fig:spatial-diagram}.

\begin{figure}[htbp]
    \centering
\begin{tikzcd}[cramped,column sep=small]
	{X_1} &&&& {X_1} &&&&&& {X_1} \\
	& {X_2} && {X_2} && {Y_2} && {Y_2} && {Y_2} \\
	&& {X_3} &&&& {Y_3} && {Y_3}
	\arrow["{f_1}"', from=1-1, to=2-2]
	\arrow["{f_1}", from=1-5, to=2-4]
	\arrow["{g_1}", from=1-5, to=2-6]
	\arrow["{g_1}"', from=1-11, to=2-10]
	\arrow["{f_2}"', from=2-2, to=3-3]
	\arrow["{f_2}", from=2-4, to=3-3]
	\arrow["{g_2}", from=2-6, to=3-7]
	\arrow["{g_2}"', from=2-8, to=3-7]
	\arrow["{g_2}", from=2-8, to=3-9]
	\arrow["{g_2}"', from=2-10, to=3-9]
\end{tikzcd}
    \caption{A diagram of spaces and geometric maps encoding a complex multigrid scheme. This is a 2 step V-Cycle followed by a 2 step W-Cycle. Such diagrams are informally used in the multigrid literature. Our approach uses them as formal diagrams in the category $\mathsf{GM}$. The functor $M$ (\cref{def:mfunc}) translates such diagrams into geometric multigrid schemes.}
    \label{fig:spatial-diagram}
\end{figure}
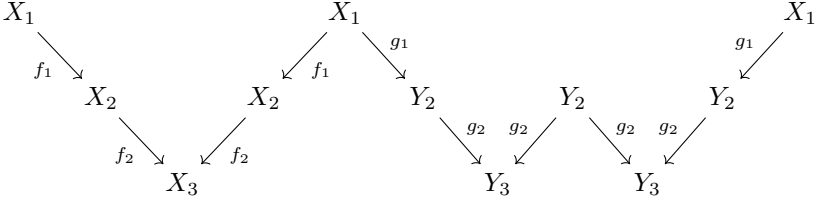

\input{csc-matrix.tex}

\subsubsection{Transferring discrete differential forms along geometric maps}

The implementation of geometric maps discussed above lends itself to a clean methodology 
for transferring vector fields along such maps. As a note, we discuss vector fields in this section; 
techniques for converting between vector fields and discrete differential forms on a fixed 
mesh are discussed in any reference on the DEC such as Section 5 of Hirani~\cite{hirani_discrete_2003}.

Transferring vector fields along geometric maps is crucial for multigrid methods, which 
involve interpolating such a quantity from a grid to a subdivided grid, and then restricting
it back to the coarse grid. We shall see that the former can be given precisely just by a multiplication 
by the matrix of the geometric map as defined above; the latter requires a bit more thought.

We exemplify the situation using the geometric map $f:V\to W$ in Figure \ref{fig:ex-subdiv}, which is an instance of one of the
built-in subdivisions used in our software.
A vector on $W$ is given by simply attaching a scalar to each vertex $w_i$. All possible such assignments form a vector field on $W$. As a special case,
a scalar function on $W$ of Figure \ref{fig:ex-subdiv} can be described by a length-3 vector $\begin{bmatrix} a&b&c\end{bmatrix}$.
We send such a field over to $V$ by interpolating the values at the adjacent vertices to
get values at $v_2,v_4,v_6$; that is, the interpolation mapping is 
\[
\begin{bmatrix} a&b&c\end{bmatrix} \mapsto \begin{bmatrix} a&\frac{a+b}{2}&b&\frac{b+c}2&c&\frac{c+a}2\end{bmatrix}
\]
which is precisely right multiplication by $f.$ Because the discrete scalar function are also stored in arrays using the same vertex ordering that is used in the global barycentric coordinates, sparse matrix vector products perform the correct interpolation operation.

Observe that we did not have to do any special calculations about interpolation to produce 
this mapping; it is applicable to the entire class of geometric maps and generalizes 
far beyond subdivision mappings. This observation for the example map $f$ 
generalizes to any map and to functions of any dimension. That is:
\begin{definition}
Given a geometric map $f:X\to Y$ and a vector field $V$ on $Y,$ the \emph{interpolation} 
of $V$ to $X$ is given as follows:
\begin{enumerate}
    \item Express $V$ as a row vector of column vectors, indexed by vertices of $Y.$ (Thus 
    if $Y$ has $n$ vertices and $V$ is a $k$-dimensional vector field, we are expressing $V$ 
    as an element of $(\mathbb{R}^k)^n.$)
    \item Multiply this row vector on the right by (the matrix representing) $f.$
\end{enumerate}
\end{definition}

Turning to restriction, we are essentially trying to push a vector field \emph{forward} along
a geometric map, which is not quite an operation that's even possible in full mathematical generality.
In multigrid proper, there are two usual approaches: 
\begin{itemize}
\item The \emph{injection} operator, 
which simply maps a value from a vertex of the fine mesh $V$ to the same value at the same vertex
of the coarse mesh $W$, forgetting all the information at vertices appearing only in $V$;
\item The \emph{full-weighting} operator, which calculates the value of the restricted 
field at some $w$ by an appropriate weighted average over values at $w$'s neighbors. 
\end{itemize}

The injection operator fits uncomfortably in our setting, as it requires that every 
vertex of $W$ be hit by exactly one vertex of $V,$  which is quite a strong constraint on 
a general geometric map, as opposed to a grid refinement in particular. We can, however, give an
operation associated to an arbitrary geometric map that specializes to the full-weighting
operator in the case of multigrid. The main observation is that, given any
$f:V\to W,$ we can take the "neighbors" of some $w\in W$ to be the vertices $v\in V$ 
such that $v$ is mapped into a simplex with $w$ as one of its vertices, 
thus such that the $(v,w)$th entry in the matrix of $f$ is nonzero. Furthermore, 
the closer this matrix entry is to $1$, the closer $f(v)$ is to $w,$ thus the more we 
should want the value of a vector at $v$ to affect the value of its pushforward at 
$w.$ We can thus construct the pushforward a field on $V$ using the average of all its values
at vertices mapping to neighbors of $w$, weighted precisely by these coefficients. 
This amounts to the following definition: 

\begin{definition}
Given a geometric map $f:X\to Y,$ the \emph{restriction operator} of $f$ is 
the operator whose matrix is given by the transpose of the row-normalization of $f.$ 
The restriction operator is applied to a vector field on $X$ just as $f$ itself
is applied to a vector field on $Y.$
\end{definition}

Thus in the case of our example $f:V\to W$ above, the restriction operator is given by 
\[\begin{bmatrix}1/2&0&0\\1/4&1/4&0\\0&1/2&0\\0&1/4&1/4\\0&0&1/2\\1/4&0&1/4\end{bmatrix}.\]
The resulting effect of taking a vector field on $V$ to the one whose value at $w$ 
is half the value at $w$ itself plus a quarter the value at each neighbor of $w$ 
is familiar from the standard full weighting operator on a 1-D binary multigrid setup.
As is standard, we row-normalize these restriction maps such that each row sums to one. This enforces the weightings to preserve constant scalar functions (up to floating point error).

Note that the subdivision method described does not implement limit surfaces such as in the subdivision exterior calculus~\cite{de_goes_subdivision_2016}. The repeated subdivision of these linear meshes do not converge to smooth surfaces in the limit in general.\footnote{For flat surfaces, linear subdivision does indeed capture the desired domain, of course.} We produce finer and finer grids that discretize the original geometry. Future work could extend this approach to implement limit surfaces, by accounting for how the changes to the geometry due to refinement affect the discretization.

Note that our approach to geometric maps allows for the arbitrary subdivision of the 2-simplices of a mesh. We have implemented binary subdivision and cubic subdivision of triangles; these split each each triangle into 4 component triangles or 9 component triangles respectively. Indeed, we will perform numerical validation using maps from both binary and cubic subdivision in Section \ref{sec:poisson-ncr}.

%% file: csc-matrix.tex
\begin{figure}[h!]
    \centering
    \includegraphics{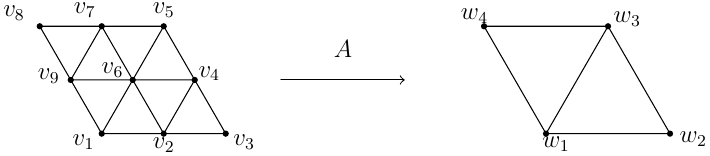}
    \[
    A=\begin{bmatrix} 
    1&1/2&0&0&0&1/2&0&0&1/2\\ 
    0&1/2&1&1/2&0&0&1/2&0&0\\
    0&0&0&1/2&1&1/2&1/2&0&0\\
    0&0&0&0&0&1/2&1/2&1&1/2
    \end{bmatrix}\]
    \[=\begin{bmatrix}1&.5&.5&.5&.5&1&.5&.5&.5&1&.5&.5&.5&.5&1&.5\\
    1&1&1&1&2&2&2&2&3&3&3&3&4&4&4&4\\
    0&1&3&4&6&7&10&13&14&16\end{bmatrix}
    \]
    \caption{A two-simplex subdivision, with its 
    matrix represented both densely and in the CSC 
    sparse representation we primarily utilize.}
    \label{fig:barycentric-csc}
\end{figure}

%% file: sec/results.tex
\section{Validation on simulation problems}
\label{sec:results}
\subsection{Poisson problem}
\label{sec:poisson-problem}

The Poisson problem serves as a classic benchmark problem, and appears as the backbone of the streamfunction-vorticity formulation of the Navier-Stokes equations and in pressure projection methods. The Poisson problem appears as a component of the porous convection problem which is to be examined in Section \ref{sec:porous-conv}. Here, we investigate the Poisson problem in its own right to demonstrate the DEC GMG\footnote{When the context is clear, this may also be abbreviated as simply GMG.} solver.

\subsubsection{Encoding the Poisson problem in the DEC}
We solve the Poisson problem, formulated in the usual vector calculus notation:
\[\Delta u = f,\]
where $u$ is an unknown scalar function and $f$ is a known scalar function. Firstly, we identify these scalar functions as differential 0-forms from the exterior calculus. That is, they are directionless quantities which can be evaluated at any point on a manifold. In the DEC, we can choose to value such 0-forms on a primal mesh or dual mesh, which are primal 0-forms and dual 0-forms respectively. Here, we will take them to be primal 0-forms.

We recover the Laplacian operator via the so-called Hodge Laplacian. In the DEC, the Hodge Laplacian operator can be constructed from a composition of the differential and codifferential operators, using function composition notation like so:
\[\Delta = \extcodif \extd + \extd \extcodif.\]
Explicitly showing the application of these operators, we can equivalently write:
\[\Delta(x) = \extcodif(\extd(x)) + \extd(\extcodif(x)),\]
where $x$ is a differential form of arbitrary degree.
The exterior derivative operator, $\extd$, generalizes the gradient operation, $\nabla$, while the codifferential operator, $\extcodif$, generalizes the divergence operation, $\nabla \cdot$. These operators have the effect of incrementing or decrementing the degree of their argument, respectively. For example, the exterior derivative of a 0-form is a 1-form (a vector-like quantity), and the codifferential of a 1-form is a 0-form. Since the codifferential of a 0-form is not defined, (there are no differential forms of negative degree), the 0-Hodge Laplacian is written without the second term:
\[\Delta = \extcodif \extd.\]
Observe that this indeed recovers the usual scalar Laplacian, (the divergence of the gradient).

The codifferential operator itself is defined in terms of the exterior derivative operator, and the Hodge star, $\star$, like so:
\[\extcodif = \sds.\]
The Hodge star operator has the effect of inverting the degree of its argument, encoding a notion of duality. For example, on a 2-manifold, the Hodge star of a 0-form is a 2-form. This captures the notion of converting a density-like quantity defined at points to a mass-like quantity defined over areas. The Hodge star of a 1-form on a 2-manifold is another 1-form. This has the effect of rotating the flow in an orthogonal direction. The inverse Hodge star operator, $\star^{-1}$, reverses this effect, but may introduce a sign change as appropriate.
Fully expanded then, we write the 0-Hodge Laplacian operator in the DEC as:
\[\Delta = \sds \extd.\]
Annotating these operators with the degree of their argument, and labeling the inverse Hodge star, we write:
\[\Delta_0 = \star_0^{-1} \tilde{\extd_1} \star_1 \extd_0.\]
In the DEC, the result of a Hodge star operation is stored on the dual mesh. We can denote this with a tilde over the exterior derivative which takes a dual form as its argument, and call it the dual exterior derivative:
\begin{equation}
\label{eq:hodge0-laplacian}
\Delta_0 = \star_0^{-1} \tilde{\extd}_1 \star_1 \extd_0.
\end{equation}

Note that in the DEC, the exterior derivative and Hodge star can be computed as matrix-vector multiplications. So, we discretize the Laplacian operator itself as the product of the matrices which encode these primitive operations. In particular, $\tilde{\extd}_1$ and $\extd_0$ can be defined as the boundary matrix of the edges of the mesh and its negated transpose, respectively. $\star_1$ can be defined as a diagonal matrix with elements the ratio of the length of each primal edge to the dual edge which crosses it in the dual mesh. $\star_0^{-1}$ can be defined as a diagonal matrix with elements the inverse of the area of each dual cell associated to a primal vertex. Figure \ref{fig:pd-scs} illustrates the cells of a primal simplicial complex and its dual. We recall that the graph Laplacian can be defined as the product of the boundary (vertex-edge incidence) matrix with its transpose, and we see this sparsity pattern occurs in the Laplacian equation. The inclusion of the Hodge star operators can thus be interpreted as accounting for metric information coming from the $n$-volumes of the $n$-simplices of the simplicial complex. In this manner, the appropriate discrete Laplacian operator is generated for any manifold-like simplicial complex.

\begin{figure}[htbp]
    \centering
    \begin{subfigure}[t]{0.45\textwidth}
        \centering
        \includegraphics[width=\textwidth]{img/primal_simplicial_complex.pdf}
        \caption{A primal simplicial complex, with primal vertices in purple, primal edges in orange, and primal triangles in blue.}
        \label{fig:p-sc}
    \end{subfigure}
    \hfill
    \begin{subfigure}[t]{0.45\textwidth}
        \centering
        \includegraphics[width=\textwidth]{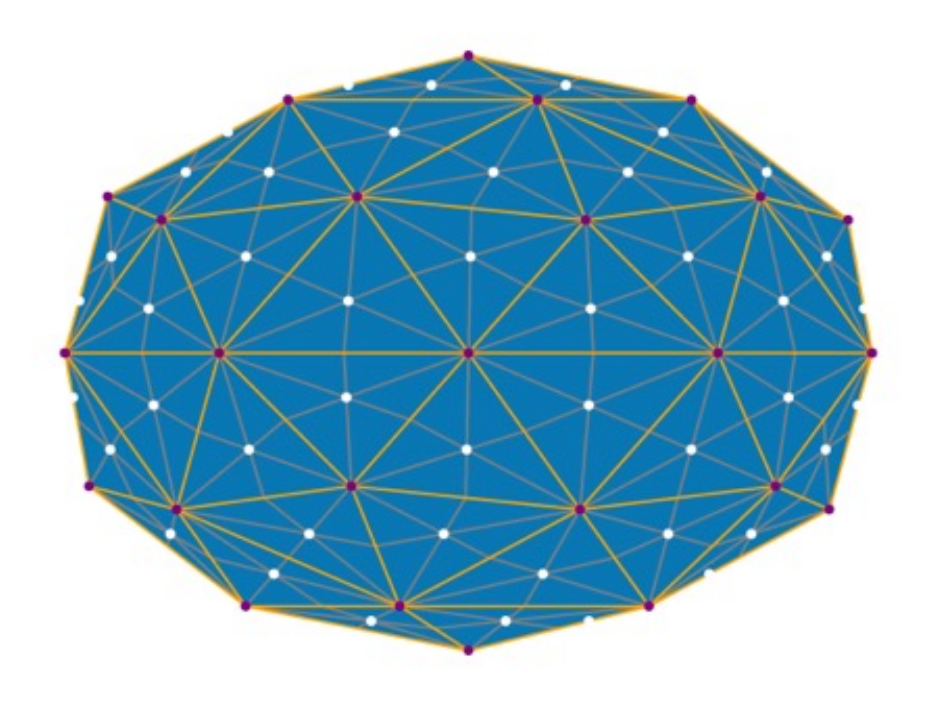}
        \caption{A dual simplicial complex, with dual vertices in white, dual edges in gray, and dual triangles the 10 or 12 component triangles surrounding a primal vertex.}
        \label{fig:d-sc}
    \end{subfigure}
    \caption{The primal simplicial complex of Figure \ref{fig:methodology-illustration-mesh} discretizing the unit sphere is shown alongside its dual mesh.}
    \label{fig:pd-scs}
\end{figure}

Solving the Poisson problem thus becomes a linear solve of the discrete Laplacian matrix. Indeed, at each level of subdivision of the mesh, we can define the discrete Laplacian matrix to represent the Laplacian operation at distinct resolutions. Note that we differ here from de Goes et al. by choosing to discretize each primitive component (both exterior derivatives and Hodge stars) of the Laplacian at each level, whereas their method involves only using discretizations of the Hodge star at each level~\cite{de_goes_subdivision_2016}. By layering these solves between our interpolation and restriction operators, we are able to perform V-cycles or W-cycles as required.

\subsubsection{Numerical convergence results}
\label{sec:poisson-ncr}

To validate the DEC GMG solver, we first perform internal validations to confirm the scaling of runtime and error is as expected. We then perform cross-comparisons of the DEC GMG solver versus other means of solving the discrete Poisson problem.

We first solve the discrete Poisson problem on (subdivisions of) the simplicial complex shown in Figure \ref{fig:poisson-mesh}. Meshes which consist of equilateral triangles are widely known to produce superior results in DEC simulations~\cite{mullen_hot_2011}. We subdivide this mesh repeatedly using 1 to 4 iterations of binary or cubic subdivision. Figure \ref{fig:poisson-mesh} also illustrates this mesh at a single level of binary subdivision.

\begin{figure}[htbp]
    \centering
    \includegraphics[width=\textwidth]{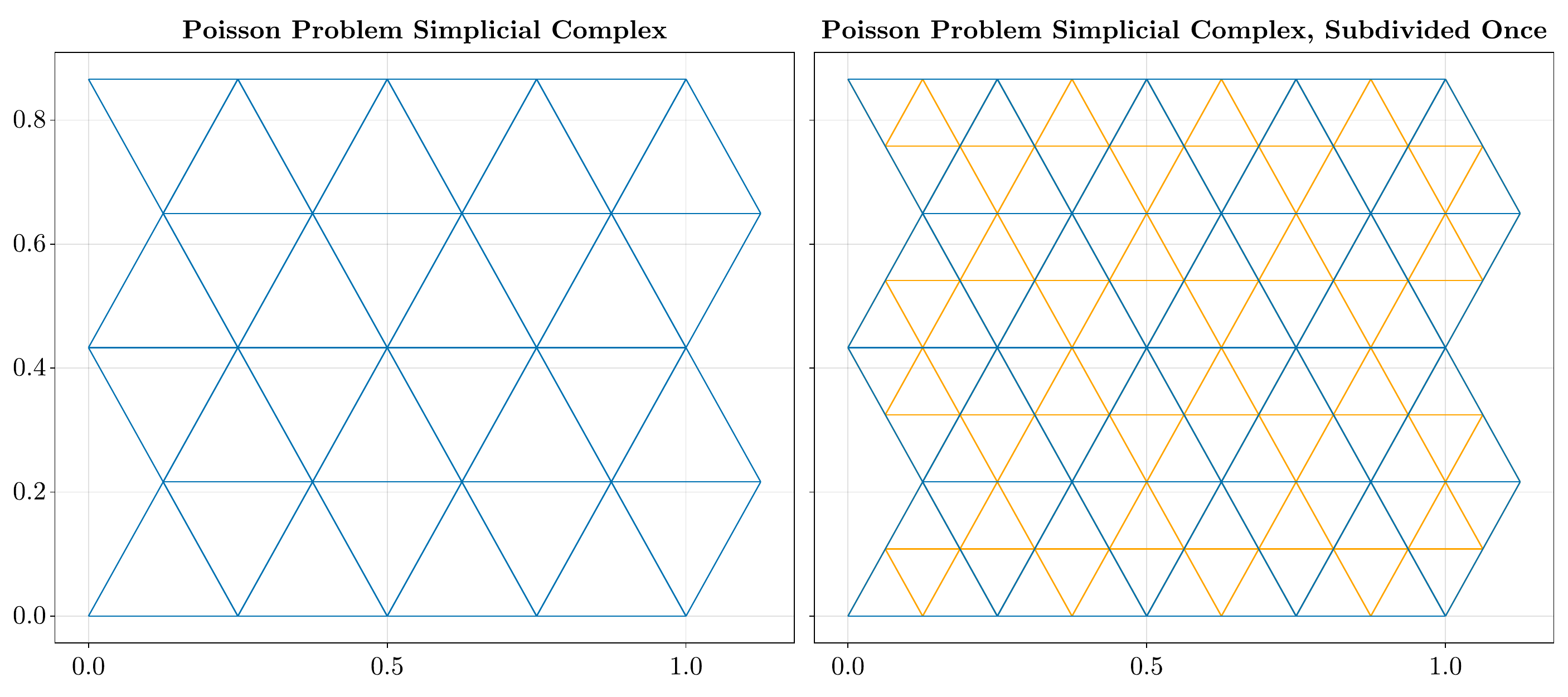}
    \caption{Left: A primal mesh consisting of 32 equilateral triangles is used as the domain for the internal validation Poisson problem test case. Right: The result of a single binary subdivision, resulting in 128 triangles.}
    \label{fig:poisson-mesh}
\end{figure}

Benchmark results of solving this Poisson problem for varying numbers of V-cycles are shown in Figure \ref{fig:vcycle-benchmarks}. The relative residual is computed via $\frac{||L u - b||_2}{||b||_2}$, where $L$ is the matrix form of the Laplacian operator, and $b$ is the known vector (on the finest mesh). 
Note that the relative residual computation should of course be expected to be more accurate on finer meshes. To ensure that this system is solvable, $b$ is initialized to $b = L r$, where $r$ is a vector of uniformly random data on $[0,1)$. We observe the runtimes of Figure \ref{fig:vcycle-benchmarks} are linear in the number of cycles, as expected, with the longest runtime belonging to the method using 4 iterations of cubic subdivision. Note that 4 iterations of cubic subdivision on the mesh of Figure \ref{fig:poisson-mesh} produces a mesh of 105,625 vertices, 25x the number of vertices in a mesh produced via 4 iterations of binary subdivision. The relative residuals of Figure \ref{fig:vcycle-benchmarks} decrease exponentially with the number of cycles, up to a point at which they level off. We observe that the rate at which the residual decreases is roughly equivalent for each scenario, but those with higher resolutions continue to find lower residual solutions as more cycles are performed. In addition to the total runtimes of Figure $\ref{fig:vcycle-benchmarks}$, we also compute the time to compute one additional V-cycle for both binary and cubic cases, arranged by the number of vertices in the finest mesh produced. These times are found to be linear in number of vertices and the results are presented in Figure \ref{fig:time_per1vcycle}.

\begin{figure}[htbp]
    \centering
    \includegraphics[width=0.8192\linewidth]{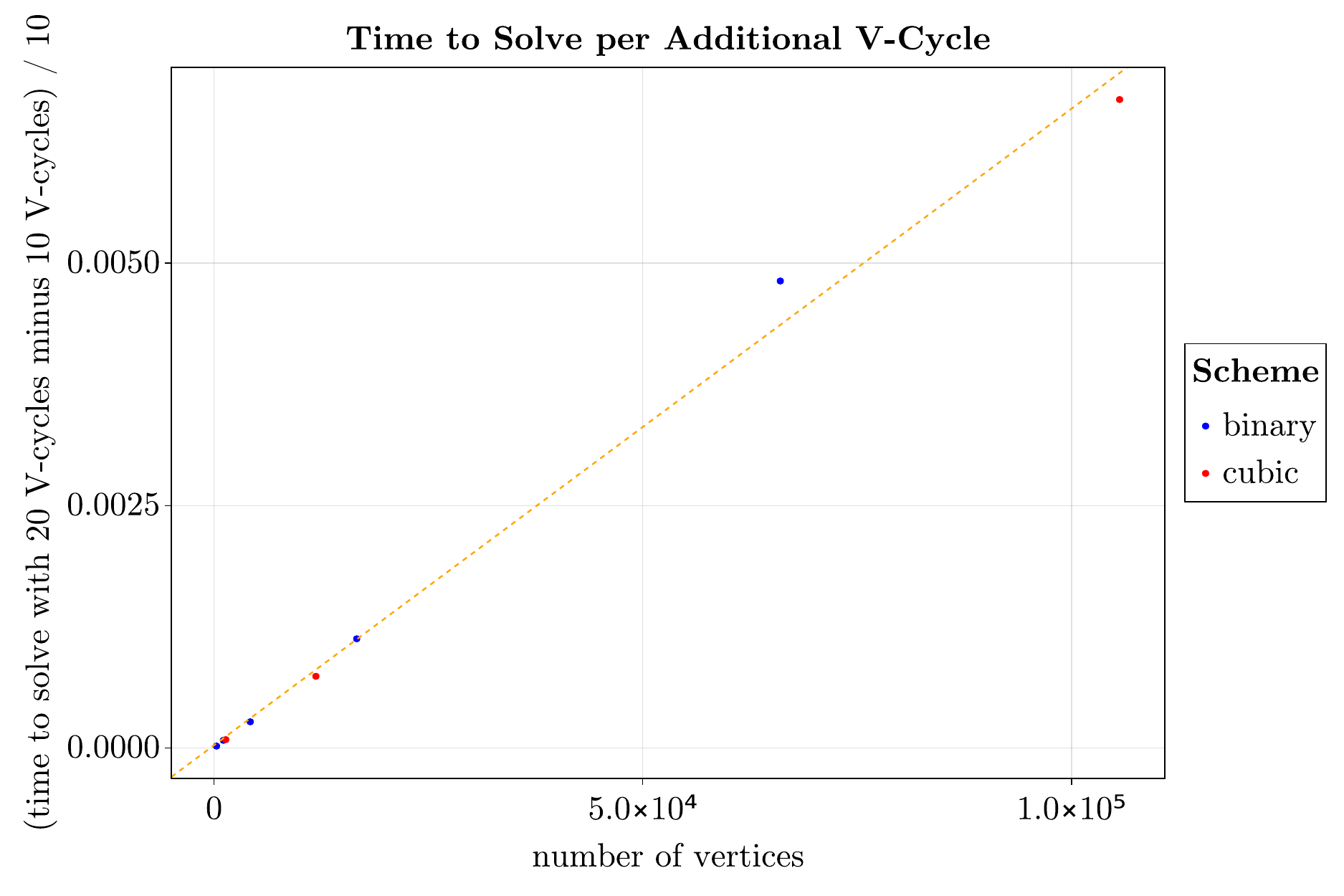}
    \caption{Shown here is the time in seconds to compute 20 V-cycles minus time to compute 10 V-cycles divided by 10, with 3 iterations of CG smoothing on each mesh. The timings are arranged by number of vertices in the finest mesh, with times associated with binary solves in blue, and cubic in red. The line of best fit to all datapoints considered together is $y = 6.57\text{e-}8x + 2.80\text{e-}5$, which corresponds to $6.57\text{e-}8$ seconds per additional vertex, with an overhead of $2.80\text{e-}5$ seconds.}
    \label{fig:time_per1vcycle}
\end{figure}

\begin{figure}[htbp]
    \centering
    \includegraphics[width=.95\textwidth]{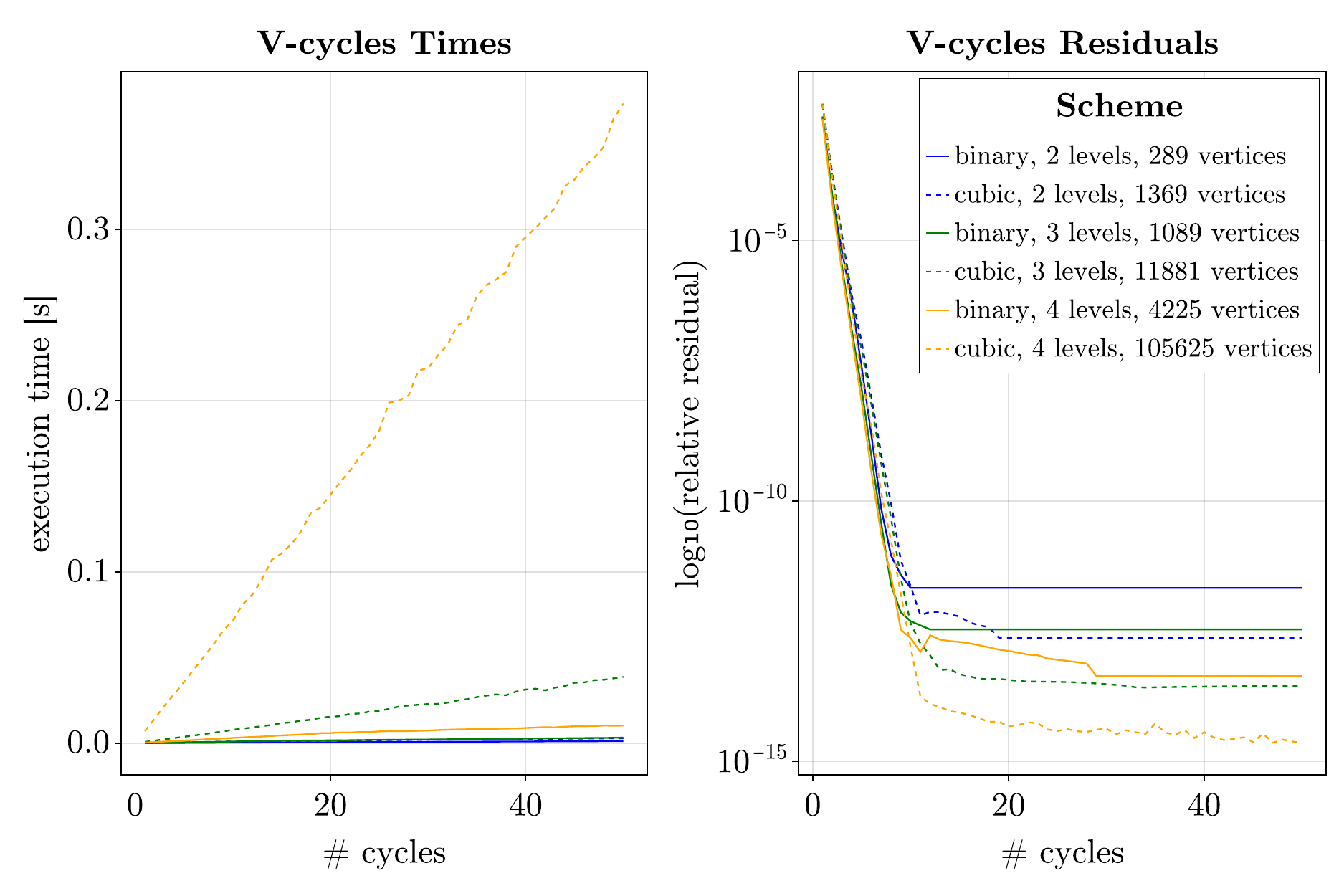}
    \caption{1 through 50 V-cycles were executed on varying subdivisions of the same domain. Their execution times and relative residuals are shown.}
    \label{fig:vcycle-benchmarks}
\end{figure}

Besides these internal validations, we also perform cross-comparison of the DEC GMG solver against other methods. These include simply directly factorizing $L$ and solving, the conjugate gradient (CG) method with ILU0 as a preconditioner, and using the DEC GMG solver as a preconditioner to CG. Note that we do not compare against an algebraic multigrid (AMG) method, due to the discrete Laplacian, in this case, being asymmetric as a result of employing reflective boundary conditions. Visualizations of the reference direct LU solution, as well as differences between the results of these solutions are shown in Figure \ref{fig:poisson_solution_comparisons}. This analysis is performed on subdivisions of the mesh in Figure \ref{fig:cross-comparison-mesh}. The discrete Laplacian $L$ considered by all methods is that defined on the finest mesh considered i.e., Figure \ref{fig:cross-comparison-mesh} subdivided 6 times by binary subdivision. In this setup, the DEC GMG methods all share the same finest mesh, and 5 w-cycles are performed, smoothing with Gauss-Seidel at each level. When the DEC GMG solver is used as a preconditioner, 2 w-cycles are performed, again with Gauss-Seidel smoothing at each level. For the finest discrete Laplacian $L$, the LU factorization took 15.4 seconds to compute, and the ILU0 factorization took 30.6 seconds. The direct factorization method of solving is used as a baseline against which iterative methods are to be compared, and we do not expect these methods to achieve this low relative residual of 1.83e-14. For the input data for these cross-comparisons, we employ data that is more representative of a typical solve than random data. Specifically, we use the quantity $\delta(\rho)$ as computed by the first step (initial conditions) of the porous convection problem considered in Section $\ref{sec:porous-conv}$. The formula to compute this quantity in the DEC is $\delta(\mathbf{g^\flat} \wedge (\alpha \rho_0 T))$, where the values for these parameters are the same as those detailed in $\ref{sec:pc-convergence}$. The physical interpretation of the results of this solve is thus pressure $P$.

\begin{figure}[htbp]
    \centering
    \includegraphics[width=0.7\textwidth]{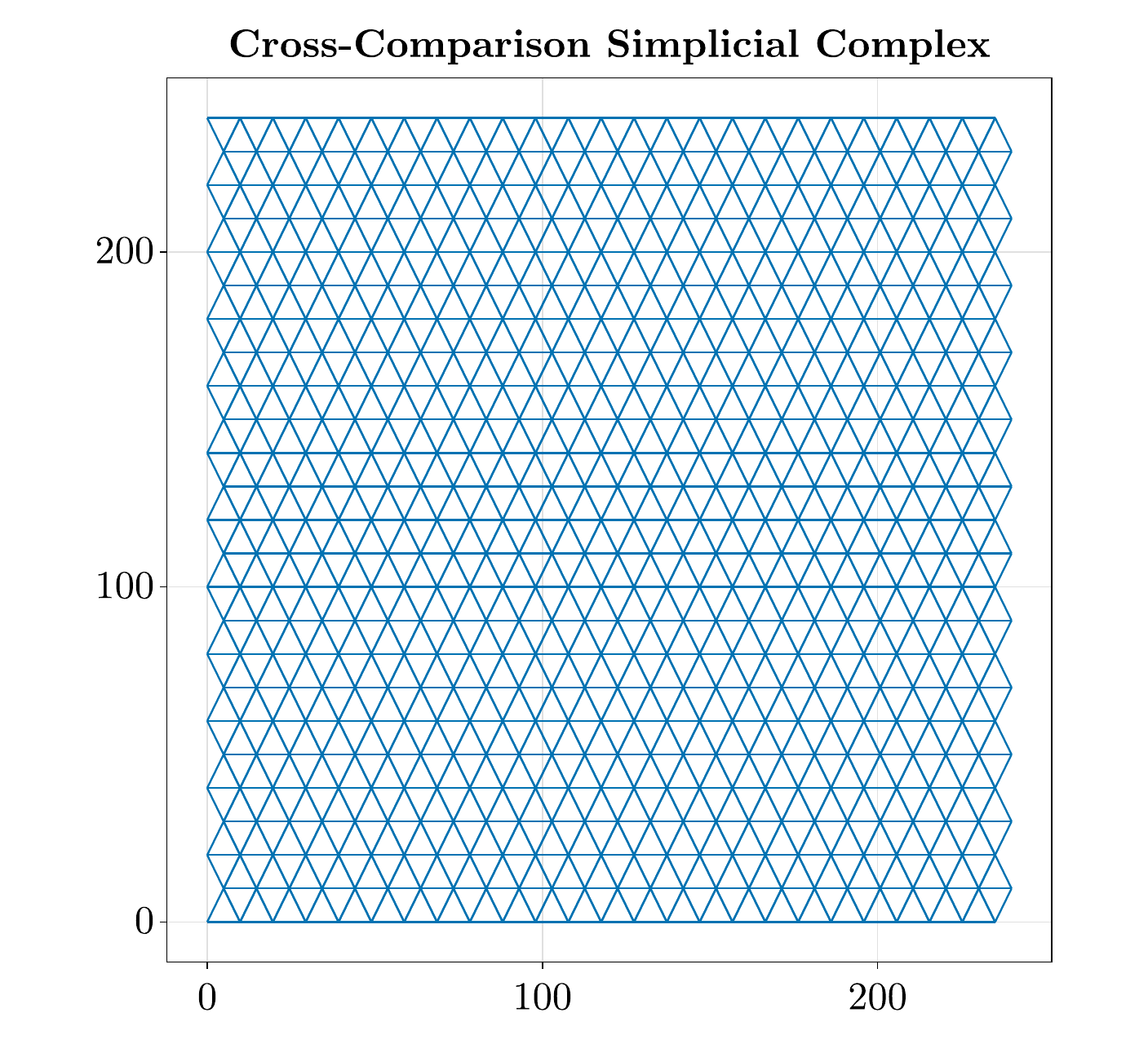}
    \caption{A primal mesh consisting of equilateral triangles is used as the domain for the cross-comparison Poisson problem test case. This coarsest mesh consists of 625 vertices, 1,776 edges, and 1,152 triangles. After 6 binary subdivisions, it consists of 2,362,369 vertices, 7,080,960 edges, and 4,718,592 triangles.}
    \label{fig:cross-comparison-mesh}
\end{figure}

\input{img/poisson_plots/poisson_divrho_composite}

\input{img/poisson_plots/poisson_divrho_data}

We observe that the relative residual of the DEC GMG solver tends to decrease as more coarsenings are considered. The residual achieved with this method decreases from 8.78e-4 to 1.10e-5 as we consider more layers. This method takes approximately 1.5 seconds to compute for all considered cases, which is favorable compared to the 13.9 seconds necessary for this configuration of CG with ILU0. However, the relative residuals of the DEC GMG solver alone are not favorable. We can improve the relative residual by instead using DEC GMG as a preconditioner for CG. Using the DEC GMG solver as a preconditioner for CG, we observe the relative residuals decrease from 9.96e-10 to 5.55e-10 as we consider more layers, superior to the 2.59e-9 obtained by CG with ILU0. Interestingly, we note that the total time to solve the system decreases, even as more levels are considered and the residual decreases, going from 38.7 to 5.98 seconds. This is due to the effectiveness of the DEC GMG solver in providing a better starting system for the CG iterative solver to consider. For completeness' sake, we also include the runtimes and relative residuals of the same experiment using random data, to be found in Table \ref{tab:poisson-comparisons-runtime-random} and Table \ref{tab:poisson-comparisons-error-random} respectively.

\input{img/poisson_plots/poisson_random_data}

\clearpage

\subsection{Porous convection}
\label{sec:porous-conv}

Following the numerical validation of the DEC GMG solver of the Poisson problem in Section \ref{sec:poisson-problem}, we now wish to study a more complex physics system using this technique. We will consider a model of porous convection which requires solving the Poisson problem to ensure a divergence-free Darcy flux. This requires a linear solve of the discrete Laplacian operator.

\subsubsection{Encoding porous convection in the DEC}

\begin{figure}[hbtp]
    \centering
    \scalebox{0.8}{\input{img/porous_convection_plots/PC_Decapode_Horizontal.tex}}
    \caption{The Porous Convection equations in the DEC are written as a diagrammatic equation, using the theory established by Patterson et al.~\cite{patterson_diagrammatic_2023} and developed in Morris et al.~\cite{morris_decapodes_2024}. This diagrammatic equation corresponds to the exterior calculus equations of Figure \ref{fig:pc_dec} with boundary conditions applied.}
    \label{fig:pc_decapode}
\end{figure}

\input{img/porous_convection_plots/pc_eqs_translation}

\begin{subequations}
Here we consider a model of porous convection translated into the DEC. A vector calculus formulation of these porous convection equations can be found in online lecture notes~\cite{rass_solving_2024}. This solver is fundamentally based on solving a Poisson problem to compute the pressure, $P$. In the exterior calculus, we encode this Poisson problem as follows:

\begin{equation}
\label{eq:poisson}
\Delta P = \delta(\mathbf{g^\flat} \wedge (\alpha \rho_0 T))
\end{equation}

where $P$ is the unknown pressure and the component quantities are to be defined shortly.

Figure \ref{fig:pc_vc} presents the vector calculus formulation of the porous convection equations, and Figure \ref{fig:pc_dec} presents each equations translation into the exterior calculus.
In the vector calculus notation, we can define the Darcy flux as in equation \ref{eq:vc_darcy},
where $T$ is the temperature, $P$ is the pressure, $q_D$ is the Darcy flux, $\mathbf{g}$ is the acceleration due to gravity, and the rest $(k, \eta, \rho_0, \alpha)$ are various constants.
Since we enforce the Darcy flux to be divergence-free, encoded as equation \ref{eq:vc_incompressible},
we can substitute equation \ref{eq:vc_darcy} into this equation and rearrange to get, equation \ref{eq:vc_divfree}.

Notice that this is the desired form of the Poisson problem, equation \ref{eq:poisson}, to solve for $P$. With $P$ determined, $q_D$ can then be calculated and the update to the temperature is defined by equation  \ref{eq:vc_pconv_heat_up}, where $\frac{\lambda}{\rho_0 c_p}$ is a diffusion coefficient.
To represent these equations in the exterior calculus, we directly convert the vector calculus operators into their exterior calculus counterparts. These operator-level translations are shown in Table \ref{tab:vc_dec_operators}. Applied to equations \ref{fig:pc_vc}, we produce the exterior calculus equations \ref{fig:pc_dec}.

\input{img/porous_convection_plots/vc_dec_operators}

In this DEC formulation, we can write the equivalent to equation \ref{eq:vc_incompressible} as Figure \ref{eq:dec_incompressible}.
For a theoretical background on differential forms for use in fluid flow applications, refer to Marsden et al.~\cite{marsden_applications_2002}.
We use the $\wedge$ operator between $\mathbf{g^\flat}$ and $T$ in equation \ref{eq:dec_darcy} because $\mathbf{g}$ is a vector field. In the vector calculus definition, equation \ref{eq:vc_darcy}, writing $\mathbf{g}T$ denotes the point-wise scaling of the gravitational acceleration vector with the temperature at that point. To accomplish that same operation in the exterior calculus, we explicitly employ the $\wedge$ operator between a $1$-form, which is the lowered $\mathbf{g^{\flat}}$, and the $0$-form $T$. The $\flat$ operator is simply a way of converting a vector field into a covector field i.e., a differential form. Note that we have elided wedge product operators between $\rho_0$ and $\alpha$, and between $\alpha$ and $\mathbf{g^{\flat}}$. Note that in the former, this wedge product is between two 0-forms, and since 0-forms are evaluated at singular points, their wedge product is equivalent to (point-wise) multiplication across the domain.

$\Lie_{q_D} T$ denotes the Lie derivative. The Lie derivative is a generalization of the directional derivative which is valid on more general manifolds. Here, this represents the change in $T$ along $q_D$. Using the notation $\uflat(\textbf{v})$  to denote the inner product of the 1-form $\uflat$ with the vector $\textbf{v}$, the Lie derivative of a 1-form can be defined via the Cartan (magic) formula for 1-forms:
\[
\Lie_v \uflat = \mathbf{v^\flat}(\extd \uflat) + \extd (\textbf{v}(\uflat)).
\]

Figure \ref{fig:pc_decapode} represents this DEC system of equations as a diagrammatic equation~\cite{patterson_diagrammatic_2023}. This diagram can be read as a computation graph, in which the values of top-level nodes are provided as initial conditions. Dotted arrows denote the arguments to a binary function - $\cdot, /, \wedge$. The ordering of arguments to binary functions is denoted by labels of 1 and 2. Intermediate variables are either named or are presented by $\bullet$. When multiple paths point to the same variable -- in this case, the $bounded \hspace{1mm} \dot{T}$ variable is pointed to by the path on the left, and the $\partial_t$ operation from above -- this denotes equality. Thus, this diagrammatic equation can be read: ``The partial derivative of $T$ with respect to time $t$ is equal to this path of operations."

We can solve these equations in the DEC by using the method of lines, discretizing our differential operators of equations \ref{eq:dec_divfree}, \ref{eq:dec_darcy}, and \ref{eq:dec_pconv_heat_up} along the spatial coordinates, while leaving the derivative with respect to time continuous in $\frac{\partial T}{\partial t}$. In Figure \ref{fig:pc_decapode}, all operations are either arithmetic or come from the DEC, save for the $bc$ operation. This is a masking operation which applies 0 to the finite updates to $T$ at the top and bottom boundaries of the domain. The \texttt{Decapodes.jl} package provides an embedded domain-specific language (eDSL or DSL) for specifying diagrammatic equations such as this. Once the DEC equations are parsed into a diagrammatic equation, or ``Decapode", this data structure is automatically compiled down to Julia code which simulates the porous convection dynamics. The DEC operators are automatically determined based on a choice of mesh. In this case, we provide the mesh data structure (simplicial complex) shown in Figure \ref{fig:pc-meshes}.

Just as with the usual vector calculus formulation, we will perform a linear solve of the discrete Laplacian in equation \ref{eq:dec_divfree}. In the following section, using the same discrete sparse matrix version of this Laplacian, we will compare the effect of using the DEC GMG solver for this linear solve versus a direct LU solve acting as a reference solution, and a typical GMRES solver with an ILU0 preconditioner. These solvers are selected by choosing different bindings for the $\Delta^{-1}$ operation shown in Figure \ref{fig:pc_decapode}. The Decapodes software decouples the specification of the physics, via the eDSL or computation graph, from the specification of the geometry and the solver. The geometry can be specified via standard triangular meshing interchange formats and the solvers are imported from the SciML suite of solvers in \texttt{OrdinaryDiffEq.jl}.

\end{subequations}

\subsubsection{Numerical convergence results}
\label{sec:pc-convergence}

To investigate the performance of the DEC GMG solver on the porous convection problem, we require inputs for the domain (simplicial complex) on which to simulate, as well as initial conditions for the dynamic temperature field $T$ and the associated constants and parameters. For convenience, we choose a similar domain and similar initial conditions to a reference pedagogical  implementation~\cite{rass_solving_2024} in the traditional vector calculus.

We study these dynamics on the domain discretized in Figure \ref{fig:pc-meshes}. We note that this domain is topologically equivalent to the unit disk, and so does not require any special consideration of the harmonic components. We refer the reader to the work of Yin et al. which considers the adjustments which may be necessary for domains with more complicated topology~\cite{yin_fluid_2023}. The discrete Laplacians on the finest and second-finest domain are constructed according to equation \eqref{eq:hodge0-laplacian} and shown in Figure \ref{fig:hodge0-laplacians}.

\begin{figure}[htbp]
    \centering
    \includegraphics[width=\textwidth]{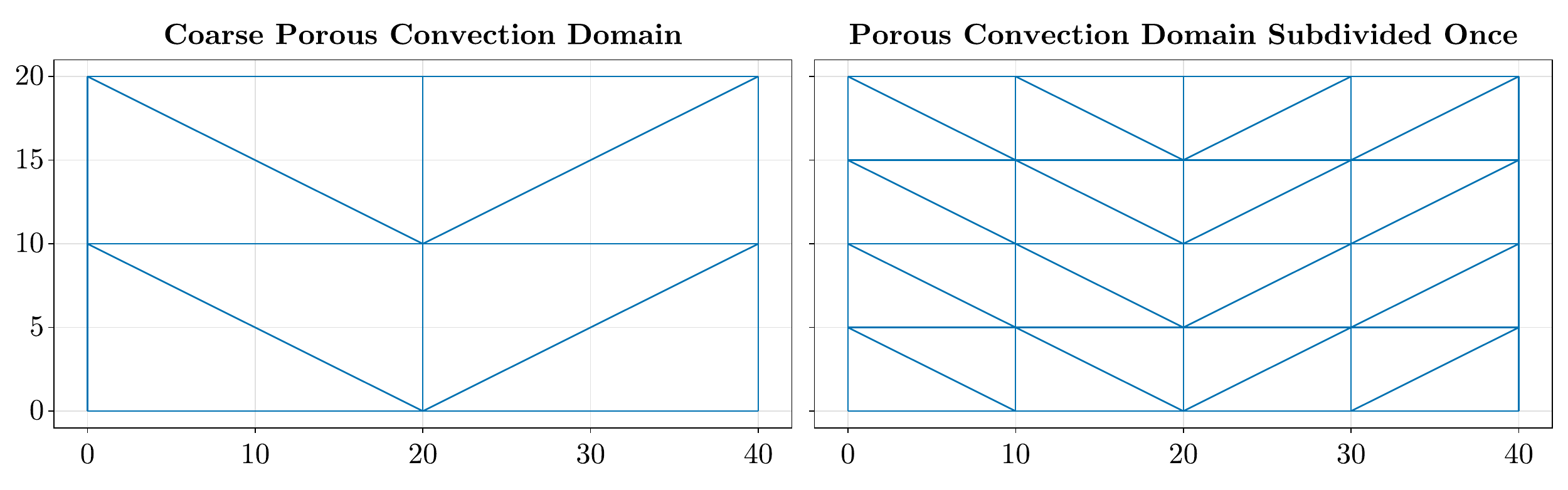}
    \caption{The Porous Convection problem is simulated on a rectangular domain. Shown here is the coarsest level used by the solver (left), and the second-coarsest level i.e., after a single binary subdivision (right). The finest domain consists of 16,641 vertices, 49,408 edges, and 32,768 triangles.}
    \label{fig:pc-meshes}
\end{figure}

\begin{figure}[htbp]
    \centering
    \includegraphics[width=.9\textwidth]{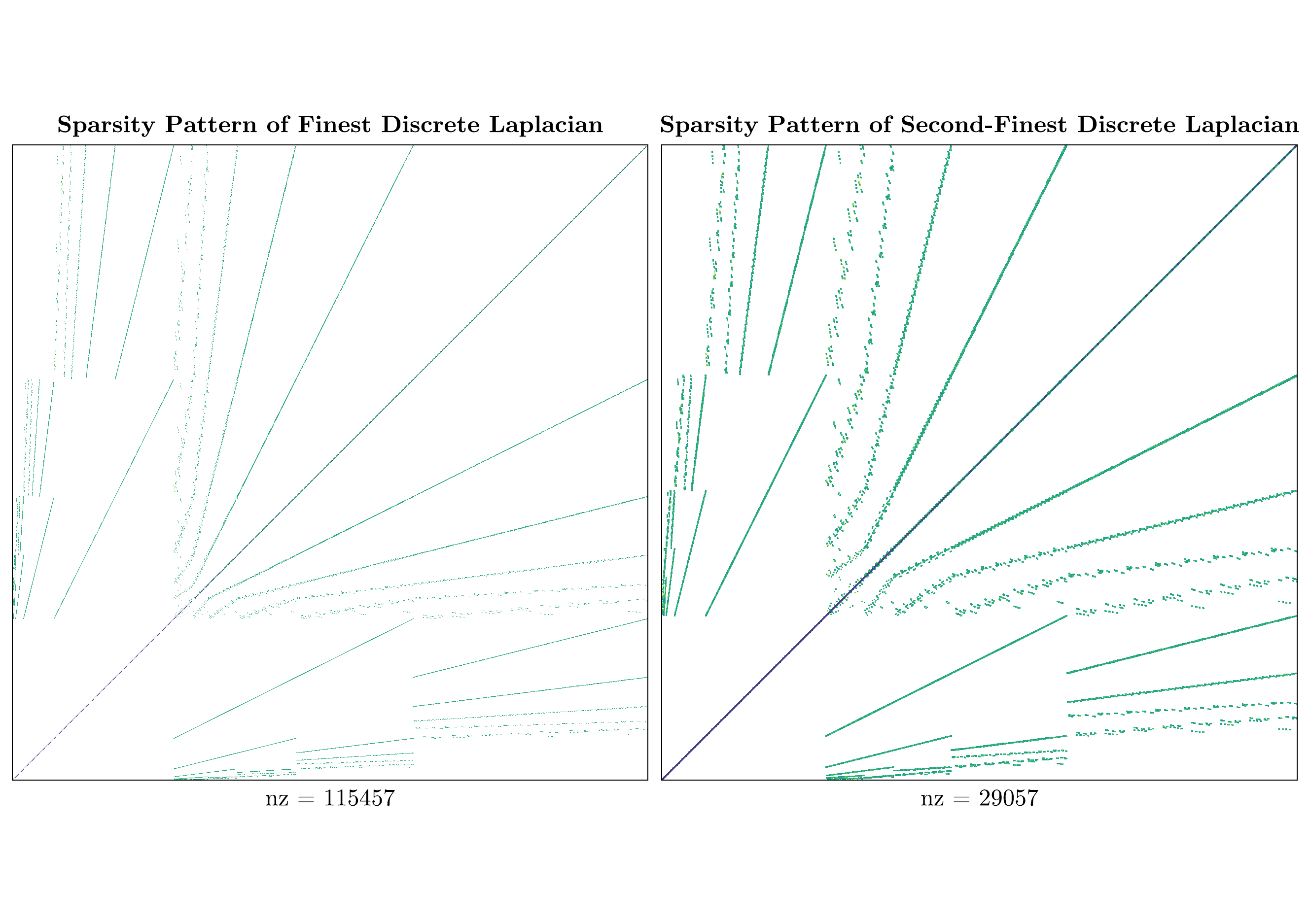}
    \caption{The finest (left) and second-finest (right) Laplacian matrices are shown above. The highest-resolution matrix is $16,641 \times 16,641$ with 115,457 nonzero entries, and the second-highest-resolution matrix is $4,225 \times 4,225$ with 29,057 nonzero entries.}
    \label{fig:hodge0-laplacians}
\end{figure}

The initial temperate $T$ for this problem is a multivariate Gaussian distribution with mean at the center of the domain, and covariance matrix $\frac{1}{2}I$, which is then multiplied by a scaling factor of 400. The resulting distribution has a maximum temperature of approximately 127$^{\circ}C$  at the center of the domain. Boundary conditions are applied which keep the top wall at a value of -100$^{\circ}C$, and the bottom wall at a value of 100$^{\circ}C$. This initial 0-form of temperature is shown in Figure \ref{fig:pc_solution_comparisons}. $g$ is set to 9.81 $m/s^2$, and this is flattened into a 1-form pointing in the negative $y$-direction. The constants $\alpha$ and $\rho_0$ are not explicitly set, but rather their product is defined to be equal to $\frac{1}{g}$. $\phi$ is set to 0.1, $Ra$ to 1000, $k_{\eta f}$ to 1, and $\Delta T$\footnote{This $\Delta T$ denotes the difference in temperature between the top and bottom walls, not the Laplacian of the temperature field.} to 200$^{\circ}C$. The product $\frac{\lambda}{\rho_0 c_p}$ is computed by the product $\frac{1}{Ra} \hspace{0.5mm}  g  \hspace{0.5mm} \alpha  \hspace{0.5mm} \rho_0  \hspace{0.5mm} k_{\eta f}  \hspace{0.5mm} \Delta T  \hspace{0.5mm} y_{max}  \hspace{0.5mm} \frac{1}{\phi}$, where $y_{max}$ is the length of the domain in the $y$-direction. The simulation is executed for an arbitrary $t=0.5$ final time.

\input{img/porous_convection_plots/pc_main_composite}

We simulate this model using a direct LU solve, GMRES with an ILU0 preconditioner provided by Krylov.jl~\cite{montoison_krylovjl_2023}, and our direct GMG solve using the Gauss-Seidel smoother provided by IterativeSolvers.jl\footnote{\url{https://github.com/JuliaLinearAlgebra/IterativeSolvers.jl}}. These results are shown in Figure \ref{fig:pc_solution_comparisons}. Again, we note that we do not compare against an AMG method, due to the asymmetry of the discrete Laplacian. The simulation using direct LU is treated as the accurate solution due to the direct solver's extremely high accuracy relative to the iterative methods. The GMRES-ILU0 configuration serves as a typical iterative solver setup that may be employed in practice. For (adaptive) timestepping, we employ the explicit \texttt{Tsit5} solver as implemented by \texttt{DifferentialEquations.jl}~\cite{rackauckas_differentialequationsjl_2017}, with the relative tolerance parameter set to $1\text{e-}9$. The adaptive time stepper for all particular configurations happens to choose the same number of total function evaluations (7677) and accepts the same number of evaluations (1275).

\begin{table}[htbp]
\centering
\caption{The DEC GMG solver is paired with a Gauss-Seidel smoother in the porous convection testcase. This is compared against the use of a direct LU solve, and a GMRES solve with an ILU0 preconditioner. The DEC GMG solver for this particular problem is 27.3\% slower but with only 18.2\% the root mean squared error (RMSE).}
\label{tab:pc-comparisons}

\begin{tabular}{crrrr}
Solver   & Runtime (s) & RMSE vs. LU \\\midrule
LU       & 17           & 0.00        \\
GMRES    & 117         & 1.16e-4     \\
DEC GMG  & 167         & 2.54e-6     \\
\end{tabular}
\end{table}

\begin{figure}[hbtp]
    \centering
    \includegraphics[width=\linewidth]{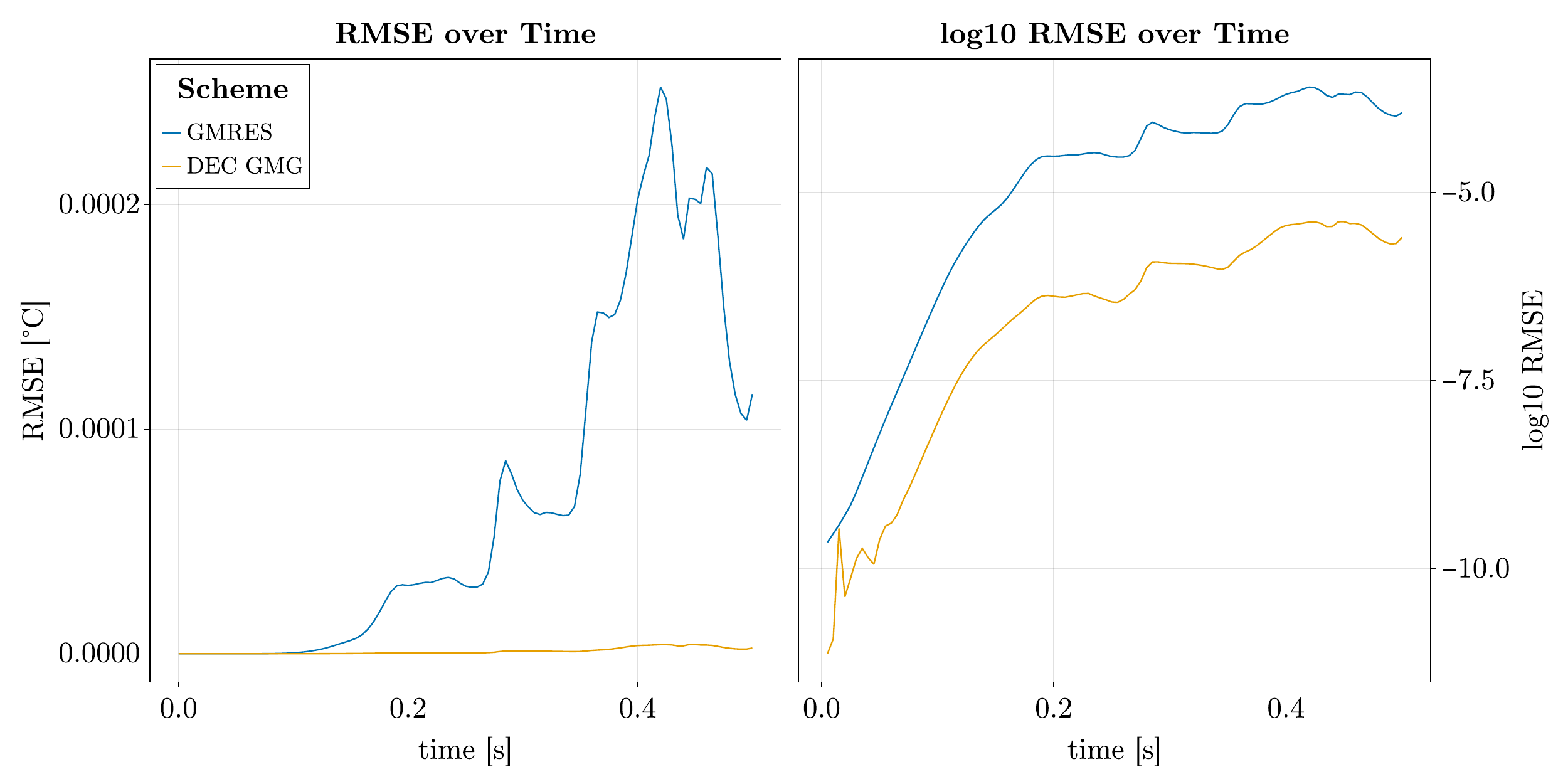}
    \caption{At each time point, the RMSE of the DEC GMG solver's solution and the GMRES solver's solution are computed (both using the direct LU solution as the reference solution). Since the RMSE at time 0 is trivially 0, the value of $-\infty$ is elided from the log plot.}
    \label{fig:rmse_over_time}
\end{figure}

We demonstrate the evolution of the RMSE 
over time for both the GMRES and the DEC GMG solves in Figure \ref{fig:rmse_over_time}.
We can see that, based on Table \ref{tab:pc-comparisons}, the GMG solver exhibits a 97.8\% decrease in RMSE 
compared to the GMRES at $t=0.5$, with an increase in runtime of 42.7\%. This demonstrates that the GMG solver's solution is good enough to remove a satisfactory amount of divergence from the pressure $P$, which would otherwise lead to non-physical behavior, and thus is able to match GMRES with ILU0 in this respect. While the suitable convergence of GMG as a solver was established in the previous section, we have now established that the GMG for this problem is a valid and viable solver for a complex physical problem.

The fact that the DEC is designed around simplicial complexes and focusing on discrete topology and geometry enables smooth integration of geometric multigrid. The modular design of the Decapodes software allows users to define additional discrete differential operators, such as the inverse Laplacian, and use them within complex multiphysics simulations. The software automations provided by Decapodes are based on the geometric nature of the DEC.

%% file: img/poisson_plots/poisson_divrho_composite.tex
\begin{figure}[htbp]
\centering
\includegraphics[width=0.99\linewidth]{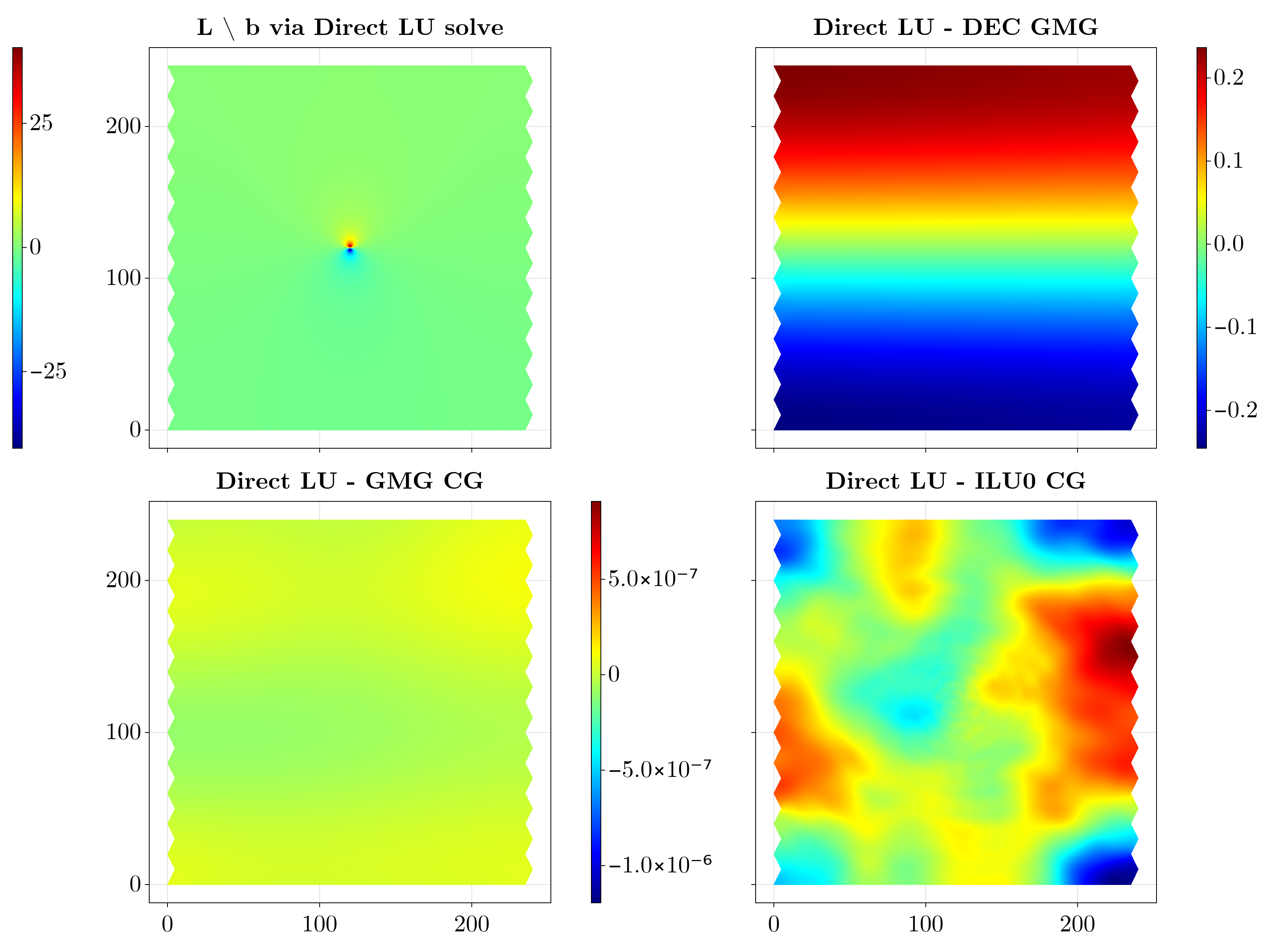}
\caption{The output of the Poisson problem direct LU solve is compared against the other three configurations. All subtractions are performed after centering all four solutions to have mean 0. Top-left: The output of the direct LU solution. In clockwise order: Top-right: The difference between the reference LU solution and the DEC GMG solution at each point. Bottom-right: The difference between the reference LU solution and the ILU0--CG solution at each point. Bottom-left: The difference between the reference LU solution and the DEC GMG--CG solution at each point. The bottom two plots share the same colorbar.}
\label{fig:poisson_solution_comparisons}
\end{figure}

%% file: img/poisson_plots/poisson_divrho_data.tex
\begin{table}[h!]

\centering
\caption{The DEC GMG solver is compared against other methods of solving the discrete Poisson problem on typical physically-meaningful data. The pre-processing times necessary to factor the discrete Laplacian via LU and ILU0 are also included for convenience. These timings were obtained on an M2 Macbook Air 2022 with 24 GB of RAM.}
\begin{tabular}{crrrrrr}
Levels & LU (s) & ILU0 (s) & LU Solve (s) & CG--ILU0 (s)   & GMG (s) & CG--GMG (s) \\\midrule
7      & 15.4   & 30.6     & 0.298        & 13.9           & 1.52    & 5.98        \\
6      & 15.4   & 30.6     & 0.298        & 13.9           & 1.50    & 10.0        \\
5      & 15.4   & 30.6     & 0.298        & 13.9           & 1.49    & 18.0        \\
4      & 15.4   & 30.6     & 0.298        & 13.9           & 1.45    & 38.7        \\
\end{tabular}
\label{tab:poisson-comparisons-runtime-divrho}
\end{table}

\begin{table}[h!]
\centering
\caption{The DEC GMG solver relative residuals are compared against other methods of solving the discrete Poisson problem on typical physically-meaningful data. Note that the Direct Factorization and CG-ILU0 methods are performed on the finest mesh, that of 6 subdivisions, but of course do not make use of coarsenings at other levels. Times for these test cases are copied to all rows for convenience.}
\begin{tabular}{crrrr}
Levels & LU Solve & CG--ILU0   & GMG      & CG--GMG   \\\midrule
7      & 1.83e-14 & 2.59e-9    & 1.10e-5  & 5.55e-10  \\
6      & 1.83e-14 & 2.59e-9    & 4.05e-5  & 3.34e-10  \\
5      & 1.83e-14 & 2.59e-9    & 1.71e-4  & 1.01e-9   \\
4      & 1.83e-14 & 2.59e-9    & 8.78e-4  & 9.96e-10  \\
\end{tabular}
\label{tab:poisson-comparisons-error-divrho}
\end{table}

%% file: img/poisson_plots/poisson_random_data.tex
\begin{table}[h!]

\centering
\caption{The DEC GMG solver is compared against other methods of solving the discrete Poisson problem on random data. The pre-processing times necessary to factor the discrete Laplacian via LU and ILU0 are also included for convenience. These timings were obtained on an M2 Macbook Air 2022 with 24 GB of RAM.}
\begin{tabular}{crrrrrr}
Levels & LU (s) & ILU0 (s) & LU Solve (s) & CG--ILU0 (s)   & GMG (s) & CG--GMG (s)   \\\midrule
7      & 15.4   & 29.9     & 0.336        & 8.33           & 1.52    & 5.30          \\
6      & 15.4   & 29.9     & 0.336        & 8.33           & 1.51    & 9.19          \\
5      & 15.4   & 29.9     & 0.336        & 8.33           & 1.48    & 14.1          \\
4      & 15.4   & 29.9     & 0.336        & 8.33           & 1.44    & 26.1          \\
\end{tabular}
\label{tab:poisson-comparisons-runtime-random}
\end{table}

\begin{table}[h!]
\centering
\caption{The DEC GMG solver relative residuals are compared against other methods of solving the discrete Poisson problem on random data. Note that the Direct Factorization and CG-ILU0 methods are performed on the finest mesh, that of 6 subdivisions, but of course do not make use of coarsenings at other levels. Times for these test cases are copied to all rows for convenience.}
\begin{tabular}{crrrr}
Levels & LU Solve & CG--ILU0   & GMG      & CG--GMG   \\\midrule
7      & 5.40e-16 & 1.33e-9    & 4.32e-8  & 2.63e-11  \\
6      & 5.40e-16 & 1.33e-9    & 1.47e-7  & 9.73e-11  \\
5      & 5.40e-16 & 1.33e-9    & 2.51e-7  & 5.63e-10  \\
4      & 5.40e-16 & 1.33e-9    & 6.48e-7  & 5.96e-10  \\
\end{tabular}
\label{tab:poisson-comparisons-error-random}
\end{table}

%% file: img/porous_convection_plots/PC_Decapode_Horizontal.tex
\begin{tikzcd}
	{\alpha \rho_0} && \bullet &&& T \\
	&& \bullet && \bullet && \bullet \\
	{\mathbf{g^\flat}} & \bullet & \rho & \bullet & \bullet & {\frac{\lambda}{\rho_0 c_p}} & \bullet \\
	\bullet && \bullet & P &&& \bullet \\
	& \bullet & {q_D} & \bullet & \bullet & \bullet & \bullet && {\dot{T}} & {bounded \hspace{1mm} \dot{T}} \\
	{-\frac{k}{\eta}} && \bullet && {-\frac{1}{\phi}} & \bullet & \bullet
	\arrow["1"{description}, dashed, from=1-3, to=1-1]
	\arrow["2"{description}, dashed, from=1-3, to=1-6]
	\arrow["{\boldsymbol{\cdot}}"', from=1-3, to=2-3]
	\arrow["{d_0}", from=1-6, to=2-5]
	\arrow["{\Delta_0}"', from=1-6, to=2-7]
	\arrow["{\partial_t}"', curve={height=-30pt}, from=1-6, to=5-10]
	\arrow["{\star_1}"', from=2-5, to=3-5]
	\arrow["2"{description}, dashed, from=3-2, to=2-3]
	\arrow["1"{description}, dashed, from=3-2, to=3-1]
	\arrow["\wedge"', from=3-2, to=3-3]
	\arrow["\delta"', from=3-3, to=3-4]
	\arrow["{\Delta^{-1}}"', from=3-4, to=4-4]
	\arrow["1"{description}, dashed, from=3-7, to=2-7]
	\arrow["2"{description}, dashed, from=3-7, to=3-6]
	\arrow["{\boldsymbol{\cdot}}", from=3-7, to=4-7]
	\arrow["2"{description}, dashed, from=4-1, to=3-3]
	\arrow["1"{description}, dashed, from=4-1, to=4-3]
	\arrow["{-}"', from=4-1, to=5-2]
	\arrow["{d_0}"', from=4-4, to=4-3]
	\arrow["{\boldsymbol{\cdot}}", from=6-3, to=5-3]
	\arrow["1"{description}, curve={height=18pt}, dashed, from=5-4, to=3-5]
	\arrow["2"{description}, dashed, from=5-4, to=5-3]
	\arrow["\wedge"', from=5-4, to=5-5]
	\arrow["{\star_0^{-1}}"', from=5-5, to=5-6]
	\arrow["2"{description}, dashed, from=5-7, to=4-7]
	\arrow["{+}"', from=5-7, to=5-9]
	\arrow["1"{description}, dashed, from=5-7, to=6-7]
	\arrow["bc", from=5-9, to=5-10]
	\arrow["2"{description}, dashed, from=6-3, to=5-2]
	\arrow["1"{description}, dashed, from=6-3, to=6-1]
	\arrow["1"{description}, dashed, from=6-6, to=5-6]
	\arrow["2"{description}, dashed, from=6-6, to=6-5]
	\arrow["{\boldsymbol{\cdot}}", from=6-6, to=6-7]
\end{tikzcd}

%% file: img/porous_convection_plots/pc_eqs_translation.tex
\begin{figure}
    \centering
    \begin{subfigure}[b]{0.44\textwidth}
    \begin{align}
        q_D = -\frac{k}{\eta}(\nabla P - \rho_0 \hspace{0.5mm} \alpha \hspace{0.5mm} \mathbf{g} \hspace{0.5mm} T) \label{eq:vc_darcy} \\
        \nabla \cdot q_D = 0 \label{eq:vc_incompressible} \\
        \nabla \cdot(\rho_0 \hspace{0.5mm} \alpha \hspace{0.5mm} \mathbf{g} \hspace{0.5mm} T) = \Delta P \label{eq:vc_divfree} \\
        \frac{\partial T}{\partial t} + \frac{1}{\phi}q_D \cdot \nabla T -\frac{\lambda}{\rho_0 c_p} \Delta T = 0 \label{eq:vc_pconv_heat_up}
\end{align}
        \caption{Porous convection as expressed in vector calculus.}
        \label{fig:pc_vc}
    \end{subfigure}
    \hspace{0.05\textwidth}
    \begin{subfigure}[b]{0.44\textwidth}
        \begin{align}
        q_D = -\frac{k}{\eta}(\extd P - (\rho_0 \hspace{0.5mm} \alpha \hspace{0.5mm}  \mathbf{g^{\flat}}) \wedge T) \label{eq:dec_darcy} \\
        \extcodif q_D = 0 \label{eq:dec_incompressible} \\
        \extcodif ((\rho_0 \hspace{0.5mm} \alpha \hspace{0.5mm}  \mathbf{g^{\flat}}) \wedge T) = \Delta P \label{eq:dec_divfree} \\
        \frac{\partial T}{\partial t} + \frac{1}{\phi}\Lie_{q_D}T -\frac{\lambda}{\rho_0 c_p} \Delta T = 0 \label{eq:dec_pconv_heat_up}
\end{align}
        \caption{Porous convection as expressed in exterior calculus.}
        \label{fig:pc_dec}
    \end{subfigure}
    \caption{The porous convection equations are given beside their direct translations into the exterior calculus.}
    \label{fig:pc_eqs}
\end{figure}

%% file: img/porous_convection_plots/vc_dec_operators.tex
\begin{table}[htbp]
\centering
\caption{Vector calculus operators are shown beside the corresponding exterior calculus operators. These translations are suitable for simulations on 2D domains.}
\begin{tabular}{ccc}
Vector Calculus Operator &  & Exterior Calculus Operator \\\midrule
$\nabla$                 &  & $\extd$                    \\
$\hspace{1mm}\nabla \cdot$    &         & $\extcodif$                \\
$\hspace{-0.7mm}\cdot \nabla$   &          & $\Lie$                     \\
\end{tabular}
\label{tab:vc_dec_operators}
\end{table}

%% file: img/porous_convection_plots/pc_main_composite.tex
\begin{figure}[htbp]
\centering
\includegraphics[width=0.8192\linewidth]{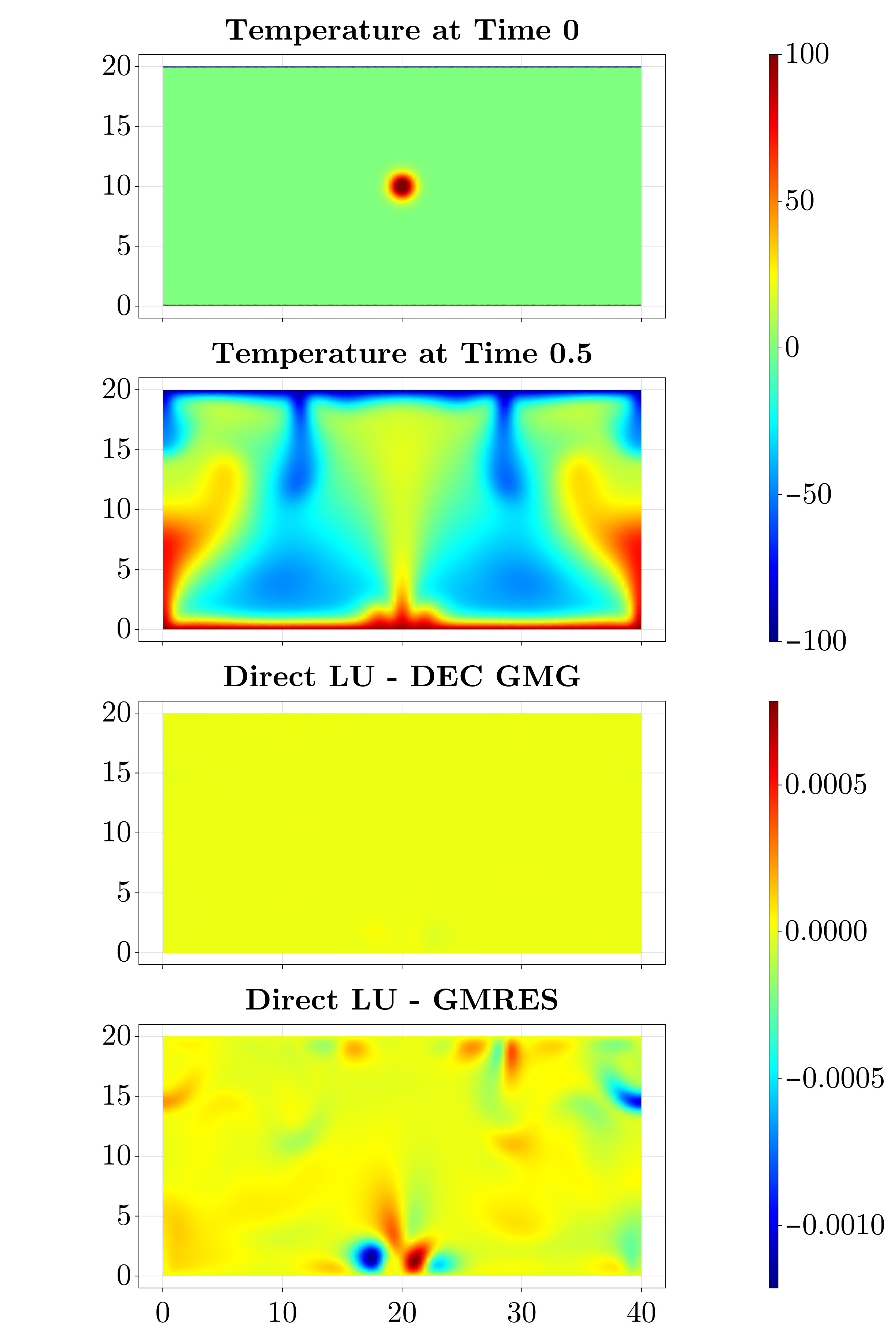}
\caption{The output of the Porous Convection problem using a direct LU solver component is compared against the solution using a DEC GMG solver and a typical GMRES solver. First Plot: The initial conditions of the temperature field. Second Plot: The temperature field at time 0.5. Third Plot: The temperature field at $t=0.5$ of the DEC GMG solution subtracted from that of the reference Direct LU solution. The root mean squared error (RMSE) of the DEC GMG solver compared to the direct LU solution is 1.90e-5. Fourth Plot: The temperature field at $t=0.5$ of the GMRES solution subtracted from that of the reference Direct LU solution. The root mean squared error (RMSE) of the GMRES solver compared to the direct LU solution is 1.04e-4. The top two and bottom two plots share the same colorbar. The initial temperature distribution is a scaled multi-variate Gaussian distribution, with top and bottom boundary conditions of -100 and +100, respectively. The maximum temperature of this initial distribution is approximately 127.}
\label{fig:pc_solution_comparisons}
\end{figure}

%% file: sec/conclusions.tex
\section{Conclusions}

We have presented an integrated physics simulation platform based on the discrete exterior calculus and a novel formulation and implementation of geometric multigrid for structure preserving discretizations that uses geometric information to provide high performance multigrid solvers for the DEC. The category theoretic constructions of general geometric maps were used in this paper to capture the interpolation and restriction maps for multiscale methods, and will serve as the basis of future work in also capturing multidomain methods such as Schwarz domain decomposition.

In a typical computational fluid dynamics problem of porous convection, this DEC GMG solver produces a solution with an 97.8\% decrease in RMSE with a 42.7\% increase in runtime compared to a typical GMRES iterative solver (2.54e-6 and 167 (s) vs. 1.16e-4 and 117 (s)). When comparing solutions of the Poisson problem itself, the DEC GMG solver serves as a useful preconditioner to the conjugate gradients method, producing a solution with a 78.6\% decrease in relative residual with a 57.0\% decrease in solving time (5.55e-10 and 5.98 (s) vs. 2.59e-9 and 13.9 (s)). Because this approach is integrated into the Decapodes multiphysics simulator system, it can be easily applied in many multiphysics models across several domains.

Future work includes the extension and application of these geometric maps to data valued directly on the higher order simplices of a simplicial complex (discrete 1-forms, 2-forms, etc.). One indirect means which this approach currently enables is using the discrete ``sharp" and ``flat" operators to interpolate and restrict discrete $1$-forms component-wise. Other future work includes extending the category theoretic formalism to capture subdivision surfaces in the vein of de Goes et al.~\cite{de_goes_subdivision_2016}.

%% file: sec/acknowledgments.tex
\section{Acknowledgments}
The authors would like to thank the attendees of the 22nd Copper Mountain Conference On Multigrid Methods for their feedback. We would like to thank Owen Lynch for fruitful discussion.
All authors were supported by the Defense Advanced Research Projects Agency, award number HR00112220038. The authors also thank the developers of Julia software packages used in this project including Krylov.jl, IterativeSmoothers.jl, Makie.jl, and the SciML ecosystem.